\documentclass{article}

\usepackage{amssymb,amsmath,latexsym,amscd,amsfonts, bbm, bm}  
\usepackage{graphicx}  
\usepackage{enumerate}
\usepackage{subcaption}
\usepackage{adjustbox}
\usepackage{authblk}

\usepackage[normalem]{ulem}
\usepackage[svgnames]{xcolor}

\usepackage{hyperref}
\hypersetup{
     colorlinks=true,           
     linkcolor=Salmon,          
     citecolor=blue,            
     filecolor=blue,            
     urlcolor=cyan,             
 }

\graphicspath{{Figs}}

\newtheorem{theorem}{Theorem}[section]

\newtheorem{definition}[theorem]{Definition} 
 
\newtheorem{claim}[theorem]{Claim} 
\newtheorem{lemma}[theorem]{Lemma}

\newtheorem{corollary}[theorem]{Corollary} 
\newtheorem{fact}[theorem]{Fact}

\newcommand{\Eq}[1]{Eq.~(\ref{#1})}
\newcommand{\Fig}[1]{Fig.~\ref{#1}}
\newcommand{\Def}[1]{Def.~\ref{#1}}
\newcommand{\Lem}[1]{Lemma~\ref{#1}}

\newcommand{\Sec}[1]{Sec.~\ref{#1}}
\newcommand{\cRef}[1]{Ref.~\cite{#1}}
\newcommand{\cRefs}[1]{Refs.~\cite{#1}}
\newcommand{\cc}[1]{~\mbox{\cite{#1}}}

\newcommand{\Thm}[1]{Theorem~\ref{#1}}
\newcommand{\App}[1]{Appendix~\ref{#1}}

\newcommand{\qedsymb}{\hfill{\rule{2mm}{2mm}}}  
\newenvironment{proof}[1][]{\begin{trivlist}  
\item[\hspace{\labelsep}{\bf\noindent Proof#1:\/}] 
 }{\qedsymb\end{trivlist}}

\newcommand{\norm}[1]{{\| #1 \|}}  
  
\newcommand{\ket}[1]{{ |{#1} \rangle }}  

\newcommand{\av}[1]{{ \langle {#1} \rangle }}

\newcommand{\ketbra}[2]{{ |{#1} \rangle\langle {#2} | }}

\newcommand{\orderof}[1]{\mathcal{O}(#1)} 
\newcommand{\EqDef}{\stackrel{\mathrm{def}}{=}}
\newcommand{\Id}{\mathbbm{1}}

\DeclareMathOperator*{\Span}{Span}
\DeclareMathOperator*{\Spec}{Spec}
\DeclareMathOperator{\poly}{poly}

\newcommand{\Tr}{\mathrm{Tr}}

\DeclareMathOperator{\supp}{supp}

\newcommand{\mcF}{\mathcal{F}}
\newcommand{\mcE}{\mathcal{E}}
\newcommand{\mcI}{\mathcal{I}}
\newcommand{\mcH}{\mathcal{H}}

\newcommand{\mcS}{\mathcal{S}}
\newcommand{\mcK}{\mathcal{K}}
\newcommand{\mcT}{\mathcal{T}}
\newcommand{\mcD}{\mathcal{D}}

\newcommand{\eps}{\epsilon}
\newcommand{\ualpha}{{\underline{\smash{\alpha}}}}
\newcommand{\ubeta}{{\underline{\smash{\beta}}}}
\newcommand{\uzero}{{\underline{0}}}
\newcommand{\CZ}{\text{C-Z}}
\newcommand{\ue}{{\underline{e}}}
\newcommand{\ueps}{{\underline{\eps}}}
\newcommand{\tQ}{\tilde Q}

\newcommand{\BBC}{\mathbb{C}}

\newcommand{\qnorm}[1]{\| #1 \|_{Q_1}}

\title{Local quantum channels
giving rise to quasi-local Gibbs states}

\author[1]{Itai Arad}
\author[2]{Raz Firanko}
\author[2]{Omer Gurevich}

\affil[1]{Centre for Quantum Technologies, Singapore}
\affil[2]{Physics Department, Technion, Haifa 3200003, Israel\footnote{Email:  \texttt{razf680@campus.technion.ac.il}}}

\begin{document}

\date{\today}

\maketitle

\begin{abstract}
  We study the steady-state properties of quantum channels with
  local Kraus operators. We consider a large family that consists of
  general ergodic 1-local (non-interacting) terms and general
  2-local (interacting) terms.  Physically, a repeated application
  of these channels can be seen as a simple model for the
  thermalization process of a many-body system. We study its steady
  state perturbatively, by interpolating between the 1-local and
  2-local channels with a perturbation parameter $\eps$. We prove
  that under very general conditions, these states are Gibbs states
  of a quasi-local Hamiltonian. Expanding this Hamiltonian as a
  series in $\eps$, we show that the $k$-th order term corresponds
  to a $(k+1)$-local interaction term in the Hamiltonian, which
  follows the same interaction graph as the channel. We also prove a
  complementary result suggesting the existence of an interaction
  strength threshold, under which the total weight of the high-order
  terms in the Hamiltonian decays exponentially fast. For
  sufficiently small $\eps$, this implies both exponential decay of
  local correlation functions and a classical algorithm with
  quasi-polynomial runtime for computing expectation values of local
  observables in such steady states. Finally, we present numerical
  simulations of various channels that support our theoretical
  results.
\end{abstract}

\section{Introduction}
\label{sec:intro}

The study of open quantum systems governed by local dynamics has
gained significant attention over the past few decades. Such systems
are prevalent in realistic scenarios, as physical systems can never
be completely isolated from their surroundings. As a result of their
interaction with the environment, they are often subject to
non-reversible dynamics, which sends the system to a fixed
point\cc{ref:Bruer2007book}. Such states are often referred to as
\emph{steady states} or \emph{non-equilibrium steady states} (NESS).
Like ground states in closed Hamiltonian systems, steady states of
many-body open quantum systems exhibit a rich variety of phenomena
such as quantum criticality\cc{ref:SO2008PT,PRXQuantum.4.030317,
PhysRevLett.132.176601}, topological
phases\cc{ref:Zhang2023fractons, ref:Moshe2019nogo,
ref:Israelis2020identifying, ref:Dimitry2007Hall}, and
chaos\cc{Kawabatta2023chaos, ref:Kawabata2023chaos2}. They can also
be used as a resource for quantum
computation\cc{ref:Verstrate2009DSE,
ref:Gilyen2017LLL,ref:Ghas2023DSE}, and are crucial for our
understanding and modeling of the noise in quantum
computers\cc{ref:Georgopoulos2021noise,ref:wang2021noise,
ref:ilin2024-learn-ss}.

More recently, many questions and results about open systems, and
steady states in particular, originate from a fruitful exchange of
ideas with the field of quantum information and quantum computation.
A central question in this context, is the preparation of a Gibbs
state on a quantum computer, a task also known as \emph{Gibbs
sampling}. This problem has gathered much attention in recent years
(see, for example, \cRefs{ref:Verstrate2009DSE, ref:Ghas2023DSE,
ref:terhal2000problem, ref:temme2011quantum,
ref:kastoryano2016quantum, ref:brandao2017Gibbs,
ref:motta2020determining, ref:wang2021variational, lambert2023DBDSE,
ref:guo2024designing, Sannia2024dissipationas, ref:chen2023quantum,
ref:Alhambra2023, ref:Consiglio2024variational,
ref:bergamaschi2024advantage, ref:ilin2024dissipative}). In many of
these works, one attempts to approach a desired Gibbs state by
driving the system to a fixed point of a discrete, stochastic map,
which can be implemented on a quantum computer. This naturally leads to the following question:

\begin{center}
  \emph{Under what conditions does the
  steady state of a discrete, local and stochastic quantum map
  approximate a Gibbs state of a local Hamiltonian?}
\end{center}

Given a local Hamiltonian $H$ on a lattice, its Gibbs state at
inverse temperature $\beta=1/k_BT$ is given by
\begin{align}
  \rho_H \EqDef \frac{1}{Z} e^{-\beta H}, \qquad 
    Z \EqDef \Tr(e^{-\beta H}) .
\end{align}
Physically, Gibbs states describe the thermal equilibrium state of a
system described by a Hamiltonian $H$, which is in contact with a
heat bath of inverse temperature $\beta$. They can also be viewed as
quantum generalizations of Boltzmann
machines\cc{ref:chowdhury2020-BoltzmanMachine,
ref:Zoufal2021-BoltzmanMachine}, and are therefore important for
quantum machine learning. Strictly speaking, any quantum state with
a full rank can be written as a Gibbs state $\rho=e^{-H}$ for
$H\EqDef -\log(\rho)$.  We would therefore like to understand under
which conditions does the local structure of the stochastic dynamics
imply a local Gibbs Hamiltonian, and more generally, what is the
relation between the locality properties of this map and the
locality of the Gibbs Hamiltonian. Given the ubiquity of Gibbs
states, which under very broad thermodynamical assumptions are the
natural equilibrium state of any system, it is reasonable to expect
that such local Gibbs Hamiltonians will be the results of many
stochastic local maps.

Our question is not only interesting from a theoretical point of
view, but also from a practical one. Local stochastic quantum maps
can generally be implemented on a quantum computer, and as such, if
they generate an interesting Gibbs state, they effectively
constitute a quantum algorithm for Gibbs state preparation. In this
context, there has been an ongoing effort for constructing such
algorithms by designing a local Lindbladian whose steady state is a
Gibbs state of a given local Hamiltonian\cc{ref:temme2011quantum,
ref:kastoryano2016quantum, ref:chen2023quantum,
ref:bergamaschi2024advantage, ref:chen2023efficient,
ref:ding2024efficient}. For example, in \cRefs{ref:chen2023quantum,
ref:chen2023efficient}, modified versions of the Davis
generators\cc{ref:Davies1974thermalization} introduce such 
Lindbladians, which consist of terms with a small effective support.
It may very well be that a better understanding of the link between
the local stochastic maps and the structure of the Gibbs Hamiltonian
that they generate will lead to more direct ways of constructing
such Gibbs sampling algorithms.

{~}

In this work, we use an elementary proof to identify a large family
of such maps.  We consider $2$-local Kraus maps defined on arbitrary
interaction graph and a perturbation parameter $\eps\in[0,1)$ such
that at $\eps=0$, the system is dissipative and non-interacting. We
show that under fairly weak assumptions, the steady state of the
system is the Gibbs state of a Hamiltonian $H(\eps)=\sum_{k\ge 0}
\eps^k H_k$, where $H_k$ is a sum of terms that are
$(k+1)$-\emph{geometrically} local with respect to the underlying
interaction graph.  We note that while this result is intuitive, it
is not straightforward to prove. A perturbation theory for the
steady state equation will yield an expansion of the steady state in
powers of $\eps$. This expansion will surely contain some
fingerprints of the underlying local dynamics; its first few terms
are expected to be sums of few operators with a simple tensorial
structure. However, the Gibbs Hamiltonian is the \emph{logarithm} of
this operator --- a highly non-local mapping --- and it is not clear
how the local structure of the steady state will descend to the
Hamiltonian. In particular, it is entirely not clear how local
geometry in $H$ emerges. Our proof, manages to circumvent all this
difficulty without delving into an intricate bookkeeping.

If the norm of the $H_k$ terms grows at most exponentially in $k$,
it follows that for sufficiently small $\eps$, the Hamiltonian
$H(\eps)$ is quasi-local, i.e., the weight of its $k$-local terms
decays exponentially. While we cannot prove that this is always the
case, we give numerical and analytical evidences to support that.
Using numerical simulations of 1D systems, we show that in all the
considered cases, $H(\eps)$ is always quasi-local, for \emph{all}
$\eps\in [0,1)$. In addition, we show the following analytical
result. We consider the expectation value of an arbitrary local
observable $A$ in the steady state $\rho(\eps)$, and expand it
perturbatively in $\eps$.  We obtain a ``Heisenberg-like''
expression for the expectation value of $A$ in the stead-state in
the form of a series of operators $A_0, A_1, A_2, \ldots$ such that
$\av{A} \EqDef \Tr(A\rho(\eps)) = \sum_{k=0}^\infty \eps^k
\Tr(A_k\rho_0)$, where $\rho_0$ is the steady state of the $\eps=0$
system (i.e., the non-interacting system), and the operators $A_k$
satisfy the following properties: 1) $A_0=A$ and the support of
$A_k$ is given by the support of all $k$-steps paths starting from
the support of $A$, and 2) the norm of $A_k$ grows at most as
$O(e^{6k})$. Consequently, for sufficiently small $\eps<
\eps_0\EqDef 1/e^6$, the perturbation series for $\av{A}$ converges
exponentially fast, which is consistent with $H(\eps)$ being
quasi-local. This result is independent of the underlying
interaction graph $G$. If we further assume that the system is
defined on a $D$-dimensional lattice, then for $\eps<\eps_0$, it
follows that $\av{A}$ can be approximated up to an additive error of
$\delta\cdot \norm{A}$ using a classical algorithm that runs in time
$e^{O(\log^D(1/\delta))}$. It is important to note that in the limit
$\eps\rightarrow 1$, the channel becomes strongly interacting and
can efficiently prepare ground states of frustration-free
Hamiltonians\cc{ref:Verstrate2009DSE, ref:Gilyen2017LLL}. Combined
with our second result, this suggests a phase transition in the
complexity of the channel somewhere in $[\eps_0, 1)$. We believe
that this transition point might be less universal and depend on the
underlying graph and local channels.

The structure of the rest of the paper is as follows. In
\Sec{sec:background} we provide some basic theoretical background
about quantum maps and Gibbs state, which will allow us to state and
prove our main theorem. In \Sec{sec:statement}, we describe the
family of systems that we consider, and give a formal statement of
our main results, which are then proved in \Sec{sec:proofs}. In
\Sec{sec:numerics} we present our numerical experiments that support
the locality of $H(\eps)$, and in \Sec{sec:discussion} we offer our
conclusions and outlook for future research.

\section{Background and notation}
\label{sec:background}

In this work we consider many-body quantum systems of $n$ qubits
that are located on the vertices of an interaction graph $G=(V,E)$
with $n=|V|$.  We denote the Hilbert space of the underlying system
by $\mcH = (\BBC^2)^{\otimes n}$, and the set of linear operators on
$\mcH$ by $L(\mcH)$. 

Given an operator $O\in L(\mcH)$, we shall denote its \emph{operator
norm} by $\norm{O}_\infty$, which is equal to its largest singular
value. The operator norm is equal to the Schatten norms $\norm{O}_p
\EqDef \big( \Tr|O|^p\big)^{1/p}$ with $p\to \infty$ and $|O|\EqDef
\sqrt{O^\dagger\cdot O}$. We shall also make use of the $p=1$
Schatten norm, $\norm{O}_1 = \Tr|O|$, known also as the \emph{trace
norm}, and the $p=2$ case $\norm{O}_2 = \sqrt{\Tr(O^\dagger \cdot
O)}$, which is also known as the \emph{Frobenius norm} or the
\emph{Hilbert Schmidt norm}. We let $\supp(O)$ denote the
\emph{support} of an operator $O$. This is the smallest subset of
qubits $S$ such that $O$ can be written as $O=\hat{O}_S\otimes
\Id_{rest}$ where $\hat{O}_S$ defined on the qubits in $S$ and
$\Id_{rest}$ is the identity operator on the rest of the qubits. We
say that $H\in L(\mcH)$ is a \emph{$k$-local Hamiltonian} if it can
be written as $H=\sum_\mu h_\mu$, where each $h_\mu$ is an Hermitian
operator with $\supp(h_\mu)\le k$. As Hamiltonians often model
physical interactions, we are in particular interested in
\emph{geometrically $k$-local} Hamiltonians, where the interaction
terms $h_\mu$ have a geometrically local support, i.e.,
$\supp(h_\mu)$ is a connected set in $G$. We say that $H=\sum_\mu
h_\mu$ is \emph{a geometrically quasi-local Hamiltonian} if it can
be written as $H=\sum_{k\ge 0} H_k$, where $H_k$ is geometrically
$k$-local Hamiltonian, and $\norm{H_k}_\infty$ decays exponentially
in $k$.

Given a Hamiltonian $H$, the state $\rho_H\EqDef
\frac{1}{Z}e^{-\beta H}$ with $Z=\Tr(e^{-\beta H})$ is called the
\emph{Gibbs state} of $H$ at inverse temperature $\beta$. This is a
physically important state, as it describes the thermal equilibrium
state of a system governed by a Hamiltonian $H$, in contact with a
heat bath at temperature $T=\frac{1}{k_B\beta}$. In this work, we
shall always consider Gibbs states of the form $e^{-H}$ by absorbing
$\beta$ and $1/Z$ in $H$. Consequently, for any full-rank
quantum state $\rho$, we can define $H_G \EqDef -\log\rho$ to be its
\emph{Gibbs Hamiltonian}. For a given $\rho$ it is interesting to
understand the properties of $H_G$, in particular, if it is local or
quasi-local with respect to some interaction graph $G$.

A quantum channel is a trace preserving and completely positive
(CPTP) linear map over $L(\mcH)$, which we shall usually denote by
$\mcE: L(\mcH)\to L(\mcH)$.  Given a channel $\mcE$, we use the
\emph{Hilbert-Schmidt inner-product} $\av{A,B}_{HS} \EqDef
\Tr(A^{\dagger}B)$ to define its adjoint map $\mcE^*$ by demanding
that $\av{\mcE^*(A),B}_{HS} = \av{A,\mcE(B)}_{HS}$ for all $A,B\in
L(\mcH)$. Note that the trace-preservation property of $\mcE$
translates into the condition $\mcE^*(\Id)=\Id$ for the
dual\cc{ref:Bruer2007book}. 

Of central interest for this work are fixed points of a channel
$\mcE$, which are states that satisfy $\mcE(\rho)=\rho$. It can be
shown that every quantum channel has at least one fixed point that
is a proper quantum state (see, for example, Theorem~4.24 in
\cRef{ref:Watrous2018book}). We shall refer to these fixed points as
\emph{steady states}. 

As in the case of operators, we say that a channel $\mcE$ (and more
generally a super-operator) acts locally on a subset of qubits $S$,
if it can be written as $\mcE = \hat{\mcE}\otimes\mcI_{rest}$, where
$\hat{\mcE}$ is defined on $L(\mcH_S)$ (the space of operators
defined on the set of qubits $S$) and $\mcI_{rest}$ is the identity
map on operators that act on the rest of the qubits. We say that a
channel $\mcE$ is $k$-local if it can be written as $\mcE = \sum_\mu
\mcF_\mu$, where each $\mcF_\mu$ acts locally on at most $k$ qubits.
As a typical example, consider a set of $k$-local quantum gates
$\{U_\mu\}$ and a corresponding probability distribution
$\{p_\mu\}$. Then $\mcE(\rho) \EqDef \sum_\mu p_\mu U_\mu \rho
U_\mu^\dagger$ is a $k$-local channel.

\section{Statement of the main results}
\label{sec:statement}

Before stating the results formally, we describe the general setup.
We consider a many-body system defined on $n$ qubits that sit on the
vertices of a graph $G=(V,E)$ with $|V|=n$. Our dynamics is
described by a repetitive application of a channel $\mcE_\eps$ that
depends smoothly on a parameter $\eps\in [0,1]$ until we converge to
the steady state:
\begin{align*}
  \rho_0 \to \mcE_\eps(\rho_0) \to \mcE^2_\eps(\rho_0) \to \ldots
  \to \rho_\infty(\eps) .
\end{align*}
The channel $\mcE_\eps$ is defined as follows. We first pick a set
of 1-qubit channels $\{\mcD_i\}$, where $\mcD_i$ acts locally on
qubit $i$. We further assume that the $\mcD_i$ channels are
ergodic\cc{ref:Sanz2010primitivity}, i.e., they have a unique,
full-rank fixed-point, and no other periodic points (spectrum with
modulus one). We denote these fixed points by $\rho_0^{(i)}$. Next,
we let $\{\mcF_e\}$ be a set of $2$-qubits channels, where $\mcF_e$
acts locally on the pair of qubits that are associated with the edge
$e$. We shall refer to the $\mcD_i$ channels as ``dissipators'' and
to $\mcF_e$ as ``correlators''. 

Given a parameter $\eps\in [0,1]$, we combine the dissipators and
the correlators to define a 2-qubit channel $\mcE_\eps^{(e)}$ that
acts locally on the qubits of the edge $e=(i,j)$:
\begin{align} 
\label{eq:Ee}
  \mcE_\eps^{(e)} \EqDef \frac{1}{2}(1-\eps)
    \big(\mcD_i + \mcD_j\big) + \eps \mcF_e
\end{align}
For $\eps\to 0$, it becomes a sum of two, uncorrelated 1-local
channels. For $\eps\to 1$, we get $\mcE_\eps^{(e)}\to \mcF_e$, which
is a general channel on the two qubits. Finally, we pick a
probability distribution $\{p_e\}$ over the edges of $G$, and define
a global channel $\mcE_\eps$ by
\begin{align}
\label{def:E}
  \mcE_\eps(\rho) \EqDef \sum_{e\in E} p_e \mcE^{(e)}_\eps(\rho)
    = \sum_{e=(i,j)\in E} p_e \Big(
      \frac{1}{2}(1-\eps)\big[ \mcD_i(\rho) 
      + \mcD_j(\rho)\big] + \eps \mcF_e(\rho)\Big).
\end{align}

For the sake of concreteness, we give a simple example of the above
dynamics, which can be implemented on a quantum computer.  Consider
a set of $n$ qubits arranged on a circle and choose our dynamics to
be translation invariant. We define:
\begin{align}
\label{eq:Ch-CZ}
  \mcD_i(\rho) &\EqDef \Tr_i[\rho]\otimes 
    \big(0.9\ketbra{+}{+} + 0.1\ketbra{-}{-}\big)_i , &
  \mcF_{i,i+1}(\rho) &\EqDef \CZ_{i,i+1}\cdot\rho\cdot \CZ_{i,i+1}, 
\end{align}
where $\CZ_{i,i+1}$ denotes the control-Z gate between qubits
$i,i+1$. Choosing $p_e$ to be the uniform distribution over the $n$
edges,  we get the global channel
\begin{align}
\label{eq:ChSimplified}
  \mcE_\eps(\rho) = \frac{1-\eps}{n}\sum_i \mcD_i(\rho) +
   \frac{\eps}{n}\sum_i \mcF_{i,i+1}(\rho) .
\end{align}
 
The dynamics described by a repeated application of the channel in
\Eq{def:E} can be thought of simplified model of thermalization for
an interacting system locally in contact with a thermal bath. On one
hand, the system undergoes local dissipation via the $\mcD_i$
dissipators, and on the other hand, it interacts with itself via the
$\mcF_e$ correlators. As the interactions are local, it is expected
that the steady state of the system will be described by a Gibbs
state of a local or a quasi-local Hamiltonian. Our first result
shows that, to a large extent, this is the case.

\begin{theorem}
\label{thm:quasi-local} 
  Consider a many-body quantum system of $n$ qubits as described
  above, in which the qubits sit on the vertices of the graph
  $G=(V,E)$, and let $\mcE_\eps$ be a local channel as in
  \eqref{def:E}. Then $\mcE_\eps$ has a unique, full-rank
  steady-state $\rho_\infty(\eps)$ for every $\eps\in [0,1)$, which
  \emph{defines} a \emph{Gibbs Hamiltonian} $H_G(\eps)\EqDef
  -\log\rho_\infty(\eps)$. This Hamiltonian is an analytic function
  of $\eps$, which can be written as a series
  \begin{align}
    H_G(\eps) = H_0 + \eps H_1 + \eps^2H_2 +\ldots
  \end{align}
  Moreover, for every $k$, the operator $H_k$ is a
  \emph{geometrically} $(k+1)$-local Hamiltonian with respect to the
  underlying graph $G$.
\end{theorem}
From the above theorem, the following corollary easily
follows 
\begin{corollary}
\label{cor:quasi-local} 
  If there exists a constant $\mu>0$ and a polynomial $C=\poly(n)$
  such that $\norm{H_k}_\infty\le C\cdot\mu^k$, then for every
  $\eps<\frac{1}{\mu}$, the Gibbs Hamiltonian $H_G(\eps)$ is
  quasi-local.
\end{corollary}

While we suspect that the assumption in the above corollary
generally holds, we do not have a proof of it. Nevertheless, we
provide two pieces of evidence in its support. Firstly, in
\Sec{sec:numerics}, we present a series of numerical experiments of
1D systems, which show that for generic systems,
$\eps^k\norm{H_k}_\infty$ decreases exponentially fast in $k$ for
every $\eps\in [0,1)$. Secondly, we show that an exponential
convergence in $k$ for $\eps$ below a certain threshold
$\eps_0$ \emph{does}
exist when considering a related quantity. Specifically, we consider
the expectation value of an observable $A$ in the steady state: 
\begin{theorem}
\label{thm:avA}
  Consider the system in \Thm{thm:quasi-local}, and assume in
  addition that the dissipators $\mcD_i$ are of the form
  $\mcD_i(\rho) = \Tr_i(\rho)\otimes W_i$, where $W_i$ is an
  arbitrary 1-qubit full-rank state. Let $\rho_\infty(\eps) = \rho_0
  + \eps \rho_1 + \eps^2\rho_2 + \ldots$ be the expansion of the
  steady state in $\eps$ and let $A$ be an observable with
  $|\supp(A)|=\ell$. Then there exists a sequence of operators
  $A=A_0, A_1, A_2, \ldots$ such that $\Tr(\rho_k A) = \Tr(\rho_0
  A_k)$ and therefore
  \begin{align}
  \label{eq:avA}
    \av{A}(\eps) \EqDef \Tr\big(\rho_\infty(\eps)A\big) 
      =  \sum_{k\ge 0}\eps^k \Tr(\rho_0 A_k).
  \end{align}
  Moreover, the operators $\{A_k\}$ have the following properties:
  \begin{enumerate}
    \item $\norm{A_k}_\infty
      \le e^{6(k+\ell)}\cdot\norm{A}_\infty$.
        
    \item $\supp(A_k) \subseteq \mathrm{Ball}_k(\supp(A))$, 
      where $\mathrm{Ball}_k(X)$ denotes a ball of radius $k$
      around a subset of vertices $X$ under the graph's
      metric.

    \item $A_k$ can be calculated classically 
      in time $T=e^{O(|\supp(A_k)|)}$.
  \end{enumerate}
\end{theorem}
 
An easy corollary from the above theorem is the existence of a
quasi-polynomial algorithm for the approximation of $\av{A}$ for
sufficiently small $\eps$:
\begin{corollary}
\label{cor:alg} 
  If in addition to the assumption of \Thm{thm:avA}, the graph $G$
  describes a $D$-dimensional lattice, then there exists an
  algorithm that for $\eps<\eps_0\EqDef 1/e^6$, an observable $A$
  and an error parameter $\delta>0$, calculates $\av{A}$ up to an
  additive error of $\delta\cdot e^{6\ell}\cdot\norm{A}_\infty$, and
  runs in time $T=e^{O(\ell\cdot\log^D(1/\delta))}$. In particular,
  for 1D systems, the algorithm is polynomial in $1/\delta$.
\end{corollary}

{~}

Finally, we combine the results of \Thm{thm:avA} with the techniques
used in the proof of \Thm{thm:quasi-local} to show that for sufficiently small
$\eps$, steady-state correlations between local regions (such as
$2$-point correlation functions) decay exponentially with the
graph distance.

\begin{theorem} 
\label{thm:decay-of-corr} Under the same conditions of
  \Thm{thm:avA}, let $\eps<\eps_0=1/e^6$ and
  let $\rho_\infty(\eps)$ be the steady-state of the channel
  $\mcE_\eps$. Then for any operators $A$ and $B$ supported in
  sub-regions $\supp(A)$ and $\supp(B)$ respectively, the steady-state
  covariance
  function $C(A,B)\EqDef \av{A B } - \av{A}\cdot \av{B}$
  satisfies 
  \begin{align} \label{eq:DOC}
    |C(A,B)| \le c \cdot e^{6(|\supp(A)|+|\supp(B)|)}\norm{A}\cdot\norm{B} 
      \cdot (\eps/\eps_0)^{d_{AB}} ,
  \end{align}
  where $c=\frac{d_{AB}+2}{(1-\eps/\eps_0)^2}$ and $d_{AB}$ is the
  lattice distance between $\supp(A)$ and $\supp(B)$.
\end{theorem}

A detailed discussion on the results and assumptions is
given in \Sec{sec:discussion}.

\section{Proofs of the main theorems}
\label{sec:proofs}

\subsection{Proof of \Thm{thm:quasi-local}}
\label{sec:quasi-local-proof}

To prove \Thm{thm:quasi-local}, we draw inspiration from the
multivariate expansion used to bound the AGSP Schmidt rank in
\cRef{ref:Arad20131DAL}. In both cases, the
analysis becomes much easier when moving to a \emph{multivariate}
setup: we replace the parameter $\eps$ by multiple parameters
$\{\eps_e\}_e$, one for each edge $e\in E$. Note that we recover the
original channel by setting $\eps_1=\eps_2=\ldots = \eps$. We use a
compact notation $\ueps \EqDef (\eps_1,\ldots,\eps_{|E|})$ to
describe these parameters, $|E|$ being the total number of edges. Then our
channel becomes:
\begin{align}
\label{eq:ch}
  \mcE_{\ueps} (\rho) =\sum_{e=(i,j)\in E} 
    p_e\Big( \frac{1}{2}(1-\eps_e) 
      [\mcD_i (\rho) + \mcD_j (\rho)] + \eps_e \mcF_e(\rho)
      \Big).
\end{align}
To proceed, we would like to expand the steady state
$\rho_\infty(\ueps)$ and its corresponding Gibbs Hamiltonian in
terms of parameters in $\ueps$. To that aim, we let $\ue_k$ denote a
subset of $k$ edges, $\ue_k=\{e_1,e_2,\ldots,e_k\}$, possibly with
repetition. Note that the order in which the edges appear in $\ue_k$
does not matter. Given a $\ue_k$, we define a shorthand notation to
the product of all the $\eps_e$ parameters that correspond to the
edges in $\ue_k$: $\ueps^{\ue_k}\EqDef \eps_{e_1} \cdot \ldots \cdot
\eps_{e_k}$. With this notation, a multivariable expansion of the
steady state and its Gibbs Hamiltonian can be \emph{formally}
written as: 
\begin{align}
\label{eq:rho-ueps}
  &\rho_\infty(\ueps)  = \rho_0 
  + \sum_{\ue_1} \ueps^{\ue_1} \rho^{(\ue_1)}
  + \sum_{\ue_2} \ueps^{\ue_2} \rho^{(\ue_2)}  + \ldots \\
  &H_G(\ueps) = H_0 
  + \sum_{\ue_1} \ueps^{\ue_1} h^{(\ue_1)}
  + \sum_{\ue_2} \ueps^{\ue_2} h^{(\ue_2)}  + \ldots .
\label{eq:H-ueps}
\end{align} 

Noting that $H_G(\eps)$ is obtained from $H_G(\ueps)$ by setting
$\eps_1=\eps_2=\ldots =\eps$, we conclude that in order to prove the
theorem, it is sufficient to show that 
\begin{enumerate}

  \item The channel $\mcE_\ueps$ has a unique fixed point
    $\rho_\infty(\ueps)$ for every $\ueps\in [0,1)^n$, and this 
    fixed point is analytic in $\ueps$ in an open set containing
    $[0,1)^n$. Consequently, the formal expansions in
    (\ref{eq:rho-ueps}, \ref{eq:H-ueps}) converge.

  \item Given a set $\ue_k=\{e_1,e_2, 
    \ldots,e_k\}$, the corresponding operator $h^{(\ue_k)}$ is
    supported on the qubits that belong to the edges in $\ue_k$.
  
  \item $h^{(\ue_k)} \neq 0$ only when its support forms a
    connected set of vertices on $G$ (and therefore it is
    supported on at most $k+1$ qubits).
\end{enumerate}
The proof of point 1 is given by the following lemma, which is
proved in \App{sec:analyticity-proof}.
\begin{lemma}
\label{lem:analyticality}
  Let $\mcE_{\ueps}$ be the channel from
  \Thm{thm:quasi-local}.  Then the following holds:
  \begin{enumerate}
    \item \label{bul1:anal} For any $\ueps \in [0,1)^n$,
      the channel $\mcE_{\ueps}$ has a unique, full-rank
      steady-state $\rho_\infty(\ueps)$, and every other eigenvalue
      $\lambda\in\Spec(\mcE_\ueps)$ satisfies $|\lambda|< 1$.
    
    \item \label{bul2:anal} 
      The matrix entries of the steady state $\rho_\infty(\ueps)$ in
      the computational basis are analytic functions of $\ueps$ in
      an open set containing $[0,1)^n$.
  \end{enumerate}
\end{lemma}

To prove points 2,3, we fix a set $\ue_k=\{e_1,e_2,\ldots,e_k\}$,
and consider an auxiliary graph $\tilde{G} = (V,\tilde E)$ that is
obtained from $G$ by removing all the edges that are not in $\ue_k$.
Let $L$ be the number of connected components in $\tilde{G}$ and let
$G_\ell=(V_\ell,E_\ell)$, $\ell=1, \ldots, L$ denote the sub-graphs
of $\tilde{G}$ corresponding to these connected components. Note
that in some subgraphs, $E_\ell$ might be the empty set (when
$V_\ell$ contains a single vertex). Next, denote by $\tilde{\ueps}$
the set of parameters obtained from $\ueps$ after setting
$\tilde{\eps}_e=0$ for all $e\notin \ue_k$, and consider the channel
$\mcE_{\tilde \ueps}$. Such channel can be written as a sum of
channels with non-overlapping supports, which correspond to the
connected components of~$\tilde G$:
\begin{align} \label{eq:sep-channel}
  \mcE_{\tilde \ueps} = \sum_{\ell=1}^L P_\ell \mcE_\ell.
\end{align}
Above, $\mcE_\ell$ is a channel acting on the qubits in $V_\ell$
and $\{P_\ell\}_\ell$ is a probability distribution. In such case,
the steady state $\rho_{\infty}(\tilde \ueps)$ decomposes into
a product state
\begin{align}
  \rho_{\infty}(\tilde \ueps) 
    = \bigotimes_{\ell=1}^L \sigma_\ell (\tilde \ueps_\ell),
\end{align}
where $\tilde \ueps_\ell$ is the set of parameters in $G_\ell$, and
each $\sigma_{\ell} (\tilde{\ueps}_\ell)$ is the steady state of
$\mcE_\ell$. Therefore, the corresponding Gibbs Hamiltonian is
\begin{align}
\label{eq:tilde-H-G}
  H_{\tilde{G}}(\tilde{\ueps}) 
    = \sum_\ell H_\ell (\tilde{\ueps}_\ell), \qquad
    H_{\ell}(\tilde{\ueps}_\ell) 
    = -\log \sigma_\ell (\tilde{\ueps}_\ell).
\end{align}
We now make the simple observation that the perturbation expansion
coefficients $\{h^{(\ue_j)}\}$ of the Gibbs Hamiltonian 
$H_G(\ueps)$ in \Eq{eq:H-ueps} \textit{do not} depend on the
perturbation parameters $\eps_1,\eps_2,\ldots$. Therefore, the
\textit{same} coefficients must appear in the perturbation expansion
of $H_{\tilde{G}}(\tilde{\ueps})$ in \Eq{eq:tilde-H-G}, as it is
obtained from $H_G(\ueps)$ by setting $\eps_e=0$ for all
$\eps\notin\ue_k$. In particular, as $\ueps^{\ue_k} =
\tilde{\ueps}^{\ue_k}$, then $H_G (\tilde{\ueps})$ must contain the
term $\ueps^{\ue_k} h^{(\ue_k)}$. Assume first that $\ue_k$ does not
correspond to a single connected component, i.e its edges are
scattered at different connected components. There will be no single
monomial $\tilde{\ueps}_\ell$ in \Eq{eq:tilde-H-G} that contains all
the variables in $\ueps^{\ue_k}$, and therefore
$H_{\tilde{G}}(\tilde{\ueps})$ from \Eq{eq:tilde-H-G} will not
contain a term that is proportional to $\ueps^{\ue_k}$. This
necessarily implies that $h^{(\ue_k)}=0$, proving property (3). On
the other hand, if $\ue_k$ forms a connected component, then there
is a subgraph $G_\ell$ where $\ue_k\subseteq E_\ell$, and since by
construction, $\tilde{G}$ contains only the edges in $\ue_k$, we
conclude that $\ue_k = E_\ell$. Consequently, the support of
$h^{(\ue_k)}$ is contained in the support of
$H_{\ell}=-\log\sigma_\ell$, which is the support of the qubits in
$\ue_k$. This proves property (2).

\subsection{Proof of \Thm{thm:avA}}
\label{sec:avA-proof}

The proof of \Thm{thm:avA} uses a rather direct perturbation theory
approach. However, to obtain tighter bounds, we work in a product
basis of operators that diagonalizes some of the perturbation
equations. In addition, we use this basis to define an operator norm
that simplifies many of the bounds. We therefore begin by defining
the operator basis and norm and discussing few of their basic
properties, before we give the full proof.

\subsubsection{The Wauli basis}
\label{sec:Wauli-basis}

We first note that without loss of generality, we may assume that
the $W_i$ states are diagonal in the computational basis, since we
can always locally redefine the computational basis of each qubit so
that it diagonalizes $W_i$. Moreover, using the same trick, we
further assume that the eigenvector of $W_i$ with the largest
eigenvalue is $\ket{0}$. We therefore define a set of $\lambda_i\in
[0,1)$ such
\begin{align}
  W_i = \frac{1}{2}(\Id + \lambda_i Z) .
\end{align}
We then define $\lambda$ to be the maximum of all these weights:
\begin{align}
  \lambda \EqDef \max_i \lambda_i < 1.
\end{align}
Next, we use the $W_i$ states to define an operator basis that
diagonalizes the $\mcD_i$ dissipator channel. We call it the ``Wauli
basis'', and it is defined as follows:
\begin{definition} [The Wauli basis]
  For a given qubit $i$, the Wauli basis is the set of operators
  \begin{align}
    \{Q^{(i)}_0,Q^{(i)}_1,Q^{(i)}_2,Q^{(i)}_3\} 
      \EqDef \{W_i, X/2, Y/2, Z/2 \} .
  \end{align}
  In order to keep the notation clean, we shall drop the $(i)$ index
  from the basis, where it is clear context on which qubit we act.
  The Wauli basis can be naturally extended to a multiqubit setup.
  For a string $\ualpha = (\alpha_1, \alpha_2, \ldots, \alpha_n)$
  with $\alpha_i \in \{0,1,2,3\}$, define $Q_\ualpha \EqDef
  Q_{\alpha_1} \otimes Q_{\alpha_2} \otimes \ldots \otimes
  Q_{\alpha_n}$. We call $\{Q_\ualpha\}_{\ualpha}$ the Wauli basis
  of the $n$-qubits space. With a slight abuse of notation, we shall
  say that the support of the string $\ualpha$ is the set of sites
  in which $\ualpha_i\neq 0$, and write $|\ualpha|$ to represent the
  size of the support.
\end{definition}
Note that the Wauli basis elements are normalized with respect to
the trace norm, i.e., $\norm{Q_\alpha}_1 = 1$. Also note that they
diagonalize the dissipators: indeed, under the settings of
\Thm{thm:avA}, $\mcD_i$ is given by $\mcD_i(\rho) \EqDef
\Tr_i(\rho)\otimes W_i$, and therefore $\mcD_i(Q^{(i)}_0) =
Q^{(i)}_0$ and $\mcD_i(Q^{(i)}_\alpha) = 0$ for $\alpha=1,2,3$.
Finally, we see that while the Wauli basis is not orthonormal with
respect to the Hilbert-Schmidt inner product, it is still a valid
basis, as its elements are linearly independent. Moreover, we can
still make use of orthogonality by looking at its \emph{dual basis}
$\{\tQ_\ubeta\}_{\ubeta}$ that is defined by the equations
\begin{align} \label{eq:bi-orthonomrality}
  \av{\tQ_\ubeta, Q_\ualpha}_{HS} = \delta_{\ualpha,\ubeta}, \quad
  \forall \ualpha,\ubeta .
\end{align}
Explicitly, this basis is given by
\begin{definition} [The Dual-Wauli basis] \label{def:dual-Waulis}
  For the qubit $i$, the Dual Wauli basis is given by
  \begin{align}
    \{\tQ^{(i)}_0,\tQ^{(i)}_1,\tQ^{(i)}_2,\tQ^{(i)}_3\} 
      = \{\Id,X,Y,Z-\lambda_i \Id)\} .
  \end{align}
  As in the Wauli basis case, we will often drop the $(i)$ index and
  expand this basis to multiple qubits by setting $\tQ_\ualpha
  \EqDef \tQ_{\alpha_1} \otimes \tQ_{\alpha_2} \otimes \ldots
  \otimes \tQ_{\alpha_n}$.  Note that as $\tQ^{(i)}_0 = \Id$, then
  for any string $\ualpha$, it holds that $|\ualpha| =
  |\supp(\tQ_\ualpha)|$.
\end{definition}
Just as the Wauli basis diagonalizes the Dissipator $\mcD_i$
channel, its dual diagonalizes its adjoint map $\mcD^*_i$ that is
given by $\mcD^*_i(A) = \Tr_i(AW_i)\otimes \Id_i$.

Throughout our proof we shall work in the Heisenberg picture, where
we apply a series of adjoint channels on $A$. Consequently, it will
be beneficial to work with the dual Wauli basis, which diagonalizes
the adjoint dissipators $\mcD_i^*$. In addition, we shall use the
following operator norm, which is defined by the dual Wauli basis:
\begin{definition}[$Q_1$ norm]
  Given an operator $O = \sum_\ualpha c_\ualpha \tQ_\ualpha$
  expanded in the dual-Wauli basis, we define its $Q_1$ norm to be
  the $L_1$ norm of its coefficients vector, i.e., 
  \begin{align}
    \norm{O}_{Q_1} \EqDef \sum_{\ualpha} |c_\ualpha| .
  \end{align}
\end{definition}
We shall also need the following definition subspaces of
operators with a bounded support with respect to the dual Wauli basis:
\begin{definition}[The $\mcS_k$ subspaces]
\label{def:Sk}
  For any integer $k\ge 0$, we let $\mcS_k$ denote the set of
  operators spanned by $\tQ_\ualpha$ with $|\ualpha|\le k$ (i.e.,
  the set of operators spanned by dual Wauli basis with support $\le
  k$).
\end{definition}

Using the definitions above, the $Q_1$ norm
satisfies three additional properties that are described in the
following lemma:
\begin{lemma} \
\label{lem:Q_1} 
  \begin{enumerate}
    \item For any operator $O$ on $n$ qubits and the identity operator $\Id$
	  on $n'$ qubits,
      \begin{equation}
      \label{eq:normOtimesQ}
        \norm{O\otimes \Id}_{Q_1} = \norm{O}_{Q_1}
      \end{equation}
      
    \item For any $O\in \mcS_k$, 
      \begin{align}
        \label{eq:Qinfineq-1}
         \norm{O}_\infty & \le (1+\lambda)^k\cdot \norm{O}_{Q_1}.
      \end{align}
	  
	\item For every $O$ supported on $\ell$ qubits,
	  \begin{align}
        \label{eq:Qinfineq-2}
         \norm{O}_{Q_1} &\le 4^\ell\cdot \norm{O}_\infty .
      \end{align}
  \end{enumerate}
\end{lemma}
Finally, we shall need the following lemma, which upperbounds the
growth in the $Q_1$ norm due to the action of the adjoint of a
$k$-local channel:
\begin{lemma}
\label{lem:Q_1-growth}
  Let $\mcE$ be a channel acting on $k$ qubits, and let $\mcE^*$ be
  its adjoint. Then for any operator~$O$,
  \begin{align}
    \norm{\mcE^*(O)}_{Q_1} \le
     4^{k}(1+\lambda)^k \cdot \norm{O}_{Q_1} \le
    2^{3k} \cdot \norm{O}_{Q_1} .
  \end{align}
\end{lemma}

The proofs of both lemmas can be found in \App{sec:Q1-proofs}.  We
are now finally in position to prove \Thm{thm:avA}.

\subsubsection{Proof of \Thm{thm:avA}}
\label{sec:actual-proof}

We start by rewriting the channel in \Eq{def:E} as
\begin{align} \label{eq:simpler-channel}
  \mcE_\eps(\rho)
    = (1-\eps)\sum_{i\in V} q_i \mcD_i
     + \eps\sum_{e\in E} p_e \mcF_e(\rho),
\end{align}
where $q_i= \frac{1}{2}\sum_{e\ni
i} p_e$ is a probability distribution on the vertices. We then decompose
\begin{align*}
  \mcE_\eps = \mcE_0 + \eps\mcE_1, 
\end{align*}
where 
\begin{align}
  \label{def:mcE0}
  \mcE_0(\rho) &\EqDef \sum_{i\in V} q_i\mcD_i(\rho)
    =\sum_{i\in V}q_i \Tr_i(\rho)\otimes W_i  ,  \\
\mcE_1(\rho) &\EqDef \sum_{e\in E}p_e\mcF_e(\rho) 
  - \sum_{i\in V} q_i\Tr_i(\rho)\otimes W_i .
\label{def:mcE1}
\end{align}
Plugging the expansion $\rho_\infty(\eps) = \rho_0 + \eps\rho_1 +
\eps^2\rho_2 + \ldots$ into the steady-state equation
$\mcE_\eps(\rho_\infty) = \rho_\infty$, we obtain the steady-state
perturbation equations
\begin{align}
  \label{eq:perturb-0}
  \mcE_0(\rho_0) &= \rho_0 , \\
  \rho_k - \mcE_0(\rho_k) & = \mcE_1(\rho_{k-1}), \quad k=1,2,3\ldots
  \label{eq:perturb-k}
\end{align}

The first equation can be solved directly, yielding:
\begin{align}
\label{eq:rho0-formula}
  \rho_0 = \bigotimes_i \rho_0^{(i)} = \bigotimes_i W_i .
\end{align}
For higher order $\rho_k$, note that both sides of \Eq{eq:perturb-k}
are defined by \emph{trace killing} maps, and therefore they do not
determine the part in $\rho_k$ that is proportional to $\Id$.
However, since
$\Tr(\rho_0)=1$ and the full $\rho_\infty(\eps)$ must have unit
trace, then clearly $\Tr(\rho_k)=0$ for all $k>0$, which fixes the
$\Id$ contribution in each $\rho_k$. Therefore, we may define a new
super-operator $\mcK$
\begin{align}
\label{eq:perturb-k2}
  \mcK(\rho_k) \EqDef \rho_k - \mcE_0(\rho_k) +
   \Tr(\rho_k) \rho_0 
\end{align}
and replace \Eq{eq:perturb-k} with 
\begin{align*}
  \mcK(\rho_k) = \mcE_1(\rho_{k-1}) .
\end{align*}
The advantage of working with this equation is that, as we will show
shortly, $\mcK$ is an invertible map. Therefore, it gives us recursive
expression for $\rho_k$:
\begin{align}
\label{eq:recurse}
  \rho_k = (\mcK^{-1} \circ \mcE_1)(\rho_{k-1}) .
\end{align}

Next, consider our local observable $A$, whose expectation value at
the steady state can be written as 
\begin{align*}
  \av{A} = \Tr\big(A\rho_\infty(\eps)) = \Tr(A\rho_0) 
    + \eps\Tr(A\rho_1) + \ldots
\end{align*} We can use the
recursive identity \Eq{eq:recurse} and move the super-operators to
the Heisenberg picture, obtaining:
\begin{align}
\label{eq:perturb-Hk}
  \Tr(A \rho_k)
    = \Tr\big(A \cdot (\mcK^{-1}\circ\mcE_1)(\rho_{k-1})\big)
  = \Tr\big((\mcK^{-1} \circ \mcE_1)^*(A) \cdot \rho_{k-1}\big) , 
    \qquad k=1,2,\ldots
\end{align}
Therefore, defining the transition super-operator $\mcT \EqDef
\mcE_1^*\circ(\mcK^*)^{-1}$, we get the recursive equation
\begin{align*}
  \Tr(A \rho_k) = \Tr\big(\mcT(A)\rho_{k-1}\big),
\end{align*}
which lets us define a series of operators
$\{A=A_0, A_1, A_2, \ldots, A_k\}$ such that 
\begin{align}
\label{def:A_k}
  A_k =\mcT^k(A), \qquad \Tr(A_k \rho_0) = \Tr(A \rho_k) .
\end{align}
This proves \Eq{eq:avA}.

For the next part of the proof, we will find an upper-bound to
$\norm{A_k}_\infty$ by working in the dual Wauli basis, and see how
it transforms under $\mcT$. Firstly, from \eqref{def:mcE0} and
\eqref{def:mcE1}, we get
\begin{align}
\label{eq:adjE0}
  \mcE^*_0(A) &= \sum_{i\in V} q_i 
    \Tr(AW_i)\otimes \Id_i, \\
\label{eq:adjE1}
  \mcE_1^*(A) &= \sum_{e\in E} p_e \mcF_e^*(A) - \mcE_0^*(A).
\end{align}
Notice that $\{Q_\ualpha\}_\ualpha$ and $\{\tQ_\ualpha\}_\ualpha$
are the eigenbases of $\mcE_0$ and $\mcE_0^*$ respectively, with
eigenvalues: 
\begin{align}
  \mcE_0(Q_\ualpha) = \sum_{i\notin\supp(\ualpha)}q_i Q_\ualpha
  = (1-q_{\ualpha})Q_\ualpha,
  \qquad
  \mcE_0^*(\tQ_\ualpha) = (1-q_{\ualpha})\tQ_\ualpha,
\end{align}
where we defined $q_{\ualpha}\EqDef \sum_{i\in\supp(\ualpha)}q_i$
and used the fact that $\{q_\ualpha\}$ form a probability
distribution. Using the above equations, together with the
definition of $\mcK$ in \Eq{eq:perturb-k2} and the fact that
$\Tr(Q_\ualpha)=\delta_{\ualpha,\uzero}$ where
$\uzero=(0,\ldots,0)$, and that $\rho_0 = Q_\uzero$ (see
\Eq{eq:rho0-formula}), we conclude that 
\begin{align}
  \mcK(Q_\ualpha) = \begin{cases}
    q_\ualpha Q_\ualpha , & \ualpha \neq \uzero \\
    Q_\ualpha, & \ualpha = \uzero
  \end{cases} .
\end{align}
Consequently, $\mcK$ is diagonal in the Wauli basis, and as its
eigenvalues are non-zero, it is also invertible. A parallel argument
applies for its conjugate $\mcK^*$ with respect to the dual Wauli
basis, and therefore, also $\mcK^*$ is invertible and
\begin{align}
\label{eq:invadjK}
  (\mcK^*)^{-1}(\tQ_\ualpha) = \begin{cases}
    \frac{1}{q_\ualpha }\tQ_\ualpha , & \ualpha \neq \uzero \\
    \tQ_\ualpha, & \ualpha = \uzero
  \end{cases} .
\end{align}

Recalling that $\mcT = \mcE^*_1\circ (\mcK^*)^{-1}$, we can now
proceed to find how it acts on the dual Wauli basis.
This is done in the following lemma, proven in \App{sec:Q1-proofs}:
\begin{claim}
\label{clm:tau} \
  \begin{enumerate}  
    \item \label{clm:tau:1} Recall the definition of the
      $\mcS_k$ subspace in \Def{def:Sk}. Then for every $O\in
      \mcS_k$, we have $\mcT(O)\in \mcS_{k+1}$. Moreover, the
      support of $\mcT(O)$ is included in the ball of radius $1$
      around the support of~$O$.

    \item \label{clm:tau:2} For any observable $O$, it holds
      $\norm{\mcT(O)}_{Q_1} \le 130\cdot \norm{O}_{Q_1}$.
  \end{enumerate}
\end{claim}

Let us use the claim to establish the upper bound on
$\norm{A_k}_\infty$.  As $A_k = \mcT^k(A)$, then by 
Bullet~\ref{clm:tau:2}, we get:
\begin{align}
  \label{eq:kTimesQnorm}
  \norm{A_k}_{Q_1}\le 130 \norm{A_{k-1}}_{Q_1}
  \le \ldots \le 130^k \norm{A}_{Q_1} .
\end{align}
Next, we use \Lem{lem:Q_1} to transform the above
inequality into a similar inequality in the operator norm. By
assumption, $A$ is supported on $\ell$ qubits, and therefore by
Bullet~\ref{clm:tau:1} of Claim~\ref{clm:tau}, $A_k=\mcT^k(A) \in
\mcS_{k+\ell}$. Using inequality~\eqref{eq:Qinfineq-1} from
\Lem{lem:Q_1}, we get
\begin{align}
\label{eq:equivInfQ}
  \norm{A_k}_\infty \le (1+\lambda)^{k+\ell}\cdot \norm{A_k}_{Q_1}.
\end{align}
On the other hand, by inequality~\eqref{eq:Qinfineq-2} of the same
lemma, combined with the fact that $|\supp(A)|=\ell$, we get
\begin{align}
  \label{eq:equivQinf}
   \norm{A}_{Q_1} \leq 4^\ell \cdot \norm{A}_\infty .
\end{align}
Altogether, inequalities (\ref{eq:equivInfQ},~\ref{eq:equivQinf})
imply that
\begin{align}
  \norm{A_k}_\infty \le \big(130(1+\lambda)\big)^k
    \cdot\big(4(1+\lambda)\big)^\ell\cdot \norm{A}_\infty 
  \le e^{6(k+\ell)}\cdot \norm{A}_\infty .
\end{align}

To conclude the proof, let us estimate $|\supp(A_k)|$ and the
complexity of calculating $A_k$. To do that, consider the how
$\mcT(\tQ_\ualpha)$ changes the support of $\tQ_\ualpha$. By
\Eq{eq:invadjK} and the definition of $\mcT$, we know that
\begin{align*}
  \mcT(\tQ_\ualpha) = \mcE_1^*\circ(\mcK^*)^{-1}(\tQ_\ualpha)
    \propto \mcE_1^*(\tQ_\ualpha).
\end{align*}
From the expression of $\mcE_1^*$ in \Eq{eq:adjE1}, we see that as
$\mcF_e^*(\Id) = \mcE_0^*(\Id) = \Id$ (which is the completeness
property of quantum channels), then the only correlators $\mcF_e^*$
acting non-trivially on $\tQ_\ualpha$ correspond to edges that
intersect with $\supp(\ualpha)$. Therefore, we can write
\begin{align*}
  \mcE_1^*(\tQ_\ualpha) = \sum_\ubeta c_\ubeta \tQ_\ubeta, 
\end{align*}
where $c_\ubeta\neq 0$ only when $\supp(\tQ_\ubeta)$ can be obtained
from $\supp(\tQ_\ualpha)$ by adding a neighboring qubit (with
respect to $G$) to one of the qubits in $\supp(\tQ_\ualpha)$.
Geometrically, this can be viewed as if under $\mcE_1^*$ (and
therefore also under $\mcT$), the support of $\tQ_\ualpha$ spreads
by one step over the graph to every direction. This implies that
$\supp(A_k)$ is included in the union of all $k$-walks over $G$ that
started from the support of $A=A_0$. When the underlying graph $G$
represents a $D$-dimensional lattice, then $\supp(A_k)$ grows like
$O(\ell \cdot k^D)$.

To estimate the complexity of calculating $A_k$, we perform the
calculation in the dual Wauli basis. The number of basis elements 
needed to represent $A_k$ is $\orderof{4^{|\supp(A_k)|}}$, and
therefore since $A_k = \mcT(A_{k-1})$, the calculation of $A_k$ can
be reduced to the a matrix-vector multiplication over a linear space
of dimension $q=4^{|\supp(A_k)|}$, which can be done in time
$O(q^2)= O(4^{2|\supp(A_k)|}) = e^{O(|\supp(A_k)|)}$. This is also
an upper bound for the first $k-1$ steps of the calculation, and
therefore the total running time needed to calculate $A_k$ is
$e^{O(|\supp(A_k)|)}$. 

\subsection{Proof of \Thm{thm:decay-of-corr}}

The idea behind the proof is as follows. We expand the covariance
function as a power series in $\eps$, and show that the lowest
non-vanishing contribution comes from terms of degree larger than
the graph distance between the supports of $A$ and $B$. \Thm{thm:avA}
will guarantee that these high order terms decay exponentially. To
show the vanishing of the low-order terms, we move the multivariate
settings used in the proof of \Thm{thm:quasi-local}, where we assign
different $\eps_j$ to different edge interactions. Expanding the
covariance in powers of $\{\eps_j\}_j$, it follows that
non-vanishing terms arise from sequences $\eps$'s that connect $A$
to $B$. This shows that non-zero contribution will come from
monomials of degree which is larger than the distance between $A$
and $B$. These, however, drop exponentially in $\eps/\eps_0$ due to
bullet 1 of \Thm{thm:avA}.

Let us begin by considering the channel from \Thm{thm:avA}, namely
$\mcE_\eps =\sum_{e\in E}p_e\mcE_\eps^{(e)}$, and assign
different variable $\eps_e$ to each edge $\mcE_e$.
The modified channel becomes $\sum_{e\in
E}p_e\mcE_{\eps_e}^{(e)}$ as presented in \Eq{eq:ch}, whose unique
steady state $\rho_\infty(\ueps)$ (given in \Eq{eq:rho-ueps}) is
analytic in $\ueps$, as claimed in \Lem{lem:analyticality}. As a
result, the steady state covariance is analytic,
which allows us to expand it in the following manner:
\begin{align}
    C(A,B)_{\rho_\infty(\ueps)} = C_0 
      + \sum_{\ue_1} \ueps^{\ue_1} C^{(\ue_1)}
      + \sum_{\ue_2} \ueps^{\ue_2} C^{(\ue_2)} + \ldots
\end{align}
Above, we used the same notation
$\ue_k=\{e_1,\ldots,e_k\}$,
$\ueps^{\ue_k}=\eps_{e_1}\dots\eps_{e_k}$ that we used in
\Eq{eq:rho-ueps}, and $\{C^{(\ue_k)}\}$ are
scalar coefficients. Note that $C_0=0$, because when setting
$\eps=0$, the steady state becomes a product state and the
covariance vanishes.

We start by showing that the expansion consists only of terms that
connect $A$ to $B$. Fix $\ue_k=\{ e_1,\ldots,e_k\}$, and assume it
does not contain a path from $A$ to $B$. As in
\Sec{sec:quasi-local-proof}, define $\tilde G$ to be the
$\ue_k$-induced subgraph, which is achieved from $G$ by restricting
to the edges in $\ue_k$. Note that due to the assumption on $\ue_k$,
each of the disjoint connected components of $\tilde G$ intersect
with $A$ or $B$ or none of them, but not $A$ and $B$ simultaneously.
Next, as done in the proof of \Thm{thm:quasi-local}, we denote by
$\tilde{\ueps}$ the set of parameters obtained from $\ueps$ after
setting $\tilde{\eps}_e=0$ for all $e\notin \ue_k$, and consider the
reduced channel $\mcE_{\tilde \ueps}= \sum_{\ell=1}^L P_\ell
\mcE_\ell$, which is a convex sum of disjoint channels on the
components of $\tilde G$ (see \Eq{eq:sep-channel}). Then corresponding steady state $\rho_{\infty}(\tilde \ueps)$
decomposes to a tensor product of states,
\begin{align}
  \rho_{\infty}(\tilde \ueps) 
    = \bigotimes_{\ell=1}^L \sigma_\ell (\tilde \ueps_\ell) 
    = \sigma_{G_A} \otimes \sigma_{G_B} \otimes \sigma_{other},
\end{align}
where $G_A$ ($G_B$) is the union of components of $\tilde G$ that
intersect with $A$ ($B$).
This is because
$\ue_k$ does not contain a path from $A$ to $B$, $G_A$ and $G_B$ are
disjoint, which guarantees that $\rho_{\infty}(\tilde \ueps) $ is
indeed a product state. Moreover, the $\rho_{\infty}(\tilde
\ueps)$-covariance function satisfies
\begin{align}
    C(O_A,O_B)_{\rho_{\infty}(\tilde \ueps) }
    = \Tr[O_A O_B \cdot  \sigma_{G_A} 
      \otimes \sigma_{G_B} ] 
     - \Tr[O_A \sigma_{G_A}]\cdot\Tr[O_B \sigma_{G_B}] = 0 .
\end{align}
Expanding the expression above in $\tilde \ueps$, its vanishing,
together with analyticity implies that each monomial vanishes
individually, and in particular, $C^{(\ue_k)}=0$.

We conclude that the expansion $C(\ueps)=\sum_\ue \eps^{\ue}
C^{(\ue)}$ contains only terms that connect $A$ and $B$. Such terms
must be of degree which is greater than or equal to the distance
$d_{AB}$ between the supports of $A$ and $B$.

We now return to the original settings of the problem
by setting $\eps_1 = \eps_2=\dots=\eps$
and expand the covariance as a power
series in $\eps$:
\begin{align} \label{eq:C-expansion}
  C(A,B)_{\rho_{\infty}(\eps)} 
    = \av{A B}_{\rho_\infty(\eps)} 
      - \av{A}_{\rho_\infty(\eps)}\cdot \av{B}_{\rho_\infty(\eps)}
    = \sum_{k\ge d_{AB}} \eps^k C^{(k)} .
\end{align}
Expanding $\av{A B}$, $\av{A}$ and $\av{B}$ according to
\Eq{eq:avA} of \Thm{thm:avA}, we get the following
expression for $C^{(k)}$:
\begin{align} \label{eq:Ck}
    C^{(k)} = \Tr(\rho_0 (A B)_k) -\sum_{j=0}^k \Tr(\rho_0 A_j) \cdot \Tr(\rho_0 B_{k-j}) .
\end{align}
Using triangle inequality and plugging Bullet 1 of \Thm{thm:avA} to
\Eq{eq:Ck}, we achieve the following upper bound 
\begin{align*}
  |C^{(k)}| &\le |\Tr(\rho_0 (A B)_k)| 
    + \sum_{j=0}^k |\Tr(\rho_0 A_j)| 
      \cdot |\Tr(\rho_0 B_{k-j})| \\
    &\le e^{6(k+\ell_A+\ell_B)}\norm{A\otimes B} 
      + \sum_{j=0}^k e^{6(j+\ell_A)}\norm{A} 
          \cdot  e^{6(k-j+\ell_B)}\norm{B} \\
    &\le e^{6(k+\ell_A+\ell_B)}\norm{A}\cdot \norm{B} 
      +(k+1) e^{6(k+\ell_A+\ell_B)}\norm{A}\cdot\norm{B} \\ 
    &= (k+2) e^{6(k+\ell_A+\ell_B)}\norm{A}\cdot \norm{B} .
\end{align*}
Here, $\ell_A$ and $\ell_B$ denote the number of qubits in the support of $A$ and $B$, respectively.
Plugging this to the expansion \eqref{eq:C-expansion} produces the
desired decay
\begin{align*}
  |C(A,B)_{\rho_\infty(\eps)}| 
    \le \sum_{k\ge d_{AB}}  \eps^k |C^{(k)}|
  \le \norm{A}\cdot \norm{B}e^{6(\ell_A+\ell_B)}
    \sum_{k\ge d_{AB}} (k+2)\left(\frac{\eps}{e^{-6}}\right)^k .
\end{align*}
Computing the geometric sum $\sum_{k\ge d} (k+2) q^k =
q^d\frac{(d+2)+\frac{q}{1-q}}{1-q} < q^d\frac{d+2}{(1-q)^2}$ gives
\Eq{eq:DOC}.

\section{Numerical findings}
\label{sec:numerics}

In this section, we present our numerical results of computer
simulations on small systems. We considered channels in the
restricted form stated in \Thm{thm:avA} on a 1D-ring of $n=10$
qubits:
\begin{align}
  \mcE_{\eps}(\rho) \EqDef \frac{\eps}{10}\sum_{i=0}^9
      \mcF_{i,i+1}(\rho) 
    + \frac{1-\eps}{10} \sum_{i=0}^{9} \Tr_i(\rho)\otimes W_i ,
\end{align}
where we identified $i=10$ with $i=0$ for periodic boundary
conditions. The main purpose of these simulations was to examine the
decay of $\norm{H_k}$ (or, more precisely, of $\eps^k \norm{H_k}$)
as a function of $k$ for typical channels. 

In all the simulations, we initialized our qubits to the state
$\rho_0\EqDef \ketbra{0}{0}^{\otimes n}$ and repeatedly applied the
channel until it approximately converged to the steady state. We
denoted by $\rho_t$ the state of the system at step $t$, and 
assumed convergence to the steady state $\rho_{\infty}(\eps)$ when
$\norm{\rho_t-\mcE(\rho_t)}_1 \le 10^{-8}$.  Once we reached the
steady state, we calculated $H_G(\eps) =
-\log{\rho_{\infty}(\eps)}$, the Gibbs Hamiltonian of our channel. 

Having calculated $H_G(\eps)$ for various $\eps\in[0,1]$, our task
was to estimate the norm of the $H_k$ terms in statement of
\Thm{thm:quasi-local}. This, however, turned out to be a numerically
challenging task, as one needs to perform some sort of an
interpolation in $\eps$ in order to extract the $H_k$ from
$H_G(\eps)$. Such procedures usually become numerically unstable as
$k$ increases. Instead, we circumvented this difficulty by
estimating a closely related quantity. Given $H_G(\eps)$, we
expanded it in terms of Pauli strings, $H_G(\eps) = \sum_\ualpha
c_\ualpha(\eps) P_\ualpha$, and defined $\hat{H}_k(\eps)$ by
\begin{align}
\label{eq:hHat}
  \hat{H}_k(\eps) \EqDef \frac{\sum_{d(\ualpha) = k} 
    c_\ualpha(\eps) P_\ualpha}{\norm{\sum_{d(\ualpha) = 1} 
      c_\ualpha(\eps) P_\ualpha}_{\infty}}.
\end{align}
In the above formula, $d(\ualpha)$ denotes the \emph{diameter} of
the support of $\ualpha$. As we worked on a ring, it was calculated
as follows. Given a $\ualpha=(\alpha_0, \ldots, \alpha_{n-1})$, we
listed all trivial segments in the string, i.e., contiguous
locations where $\alpha_i=0$, and defined $d(\ualpha)$ to be $n=10$
minus the length of the longest segment. The denominator in the
above definition is an overall factor that guarantees that
$\norm{\hat{H}_{k=1}(\eps)}_\infty=1$. 

We note that $\eps^k H_k \neq \hat{H}_k(\eps)$. The main difference,
apart from the normalization factor and the fact $\hat{H}_k(\eps)$
is geometrically $k$-local, while $H_k$ is geometrically
$(k+1$)-local, is that $H_k$ includes only terms that are of
geometrical support of \textit{at most} $k+1$ sites whereas $\hat
H_k$ includes terms that are of geometrical support of
\textit{exactly} $k$ sites. Nevertheless, we expect that for general
systems, the \emph{scaling} of these two quantities should be the
same. In what follows we plotted $\norm{\hat{H}_k(\eps)}_\infty$ vs
$k\in \{1,\ldots, 7\}$, for different values of $\eps$ and different
models. We stopped at $k=7$ instead of $k=n=10$, because already at
$k=8$, the finite system size started to be felt.

The first dynamics we simulated was composed of random $W_i$ and
random correlators $\{\mcF_{e=(i,i+1)}\}$. The former where chosen
randomly on the Bloch sphere, while the latter were chosen as random
Kraus operators of $\mathrm{rank}=3$ (see
\cRef{ref:prosen2000random} for more details). To validate the
robustness of our results, we ran several random instances of such
channels, and observed similar results in all runs. We plot
$\norm{\hat{H}_k(\eps)}_\infty$ as function of $k$ for a single
instance of such channel in \Fig{fig:Kraus}.
\begin{figure}[ht]
\centering
  \begin{subfigure}[l]{0.48\textwidth}
      \centering
      \includegraphics[width=\textwidth]{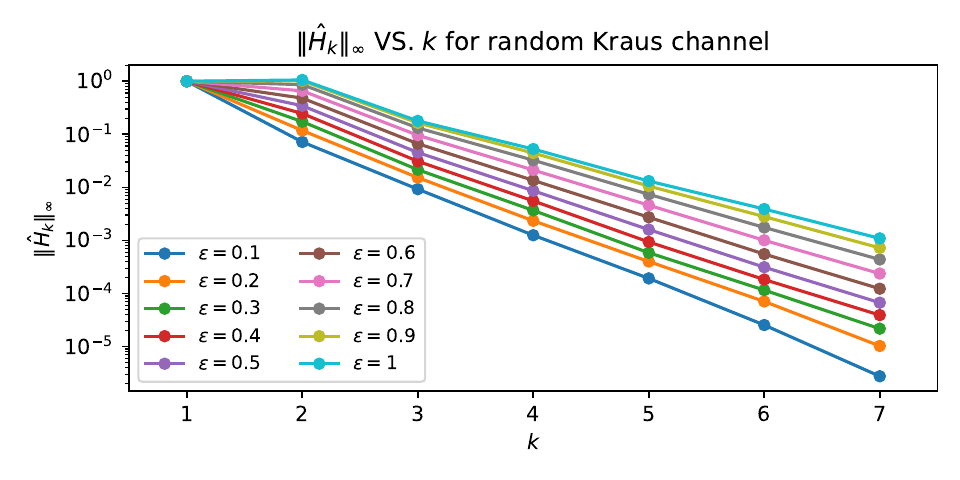}
      \caption{First model: random Kraus maps}
      \label{fig:Kraus}
  \end{subfigure}
  \hfill
  \begin{subfigure}[r]{0.48\textwidth}
      \centering
      \includegraphics[width=\textwidth]{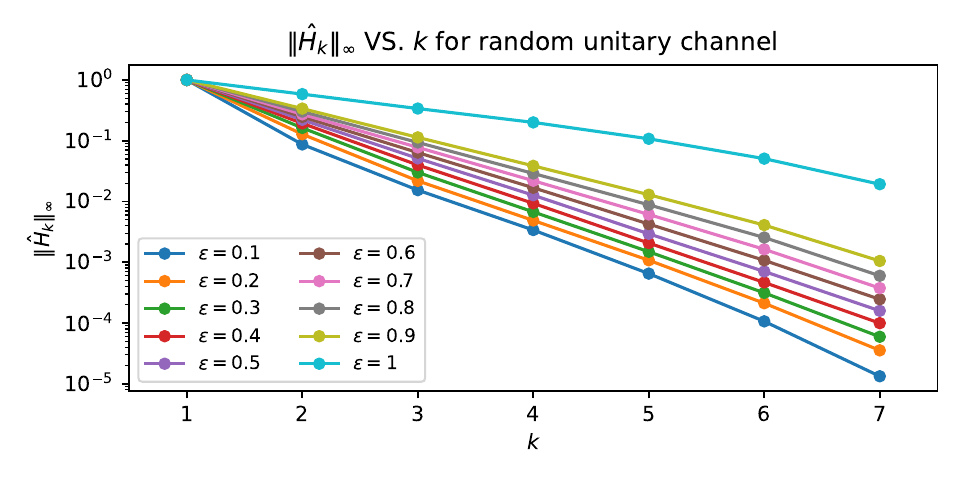}
      \caption{Second model: random unitaries ($k=1)$}
      \label{fig:U}
  \end{subfigure}
  \hfill
  \begin{subfigure}[l]{0.48\textwidth}
      \centering
      \includegraphics[width=\textwidth]{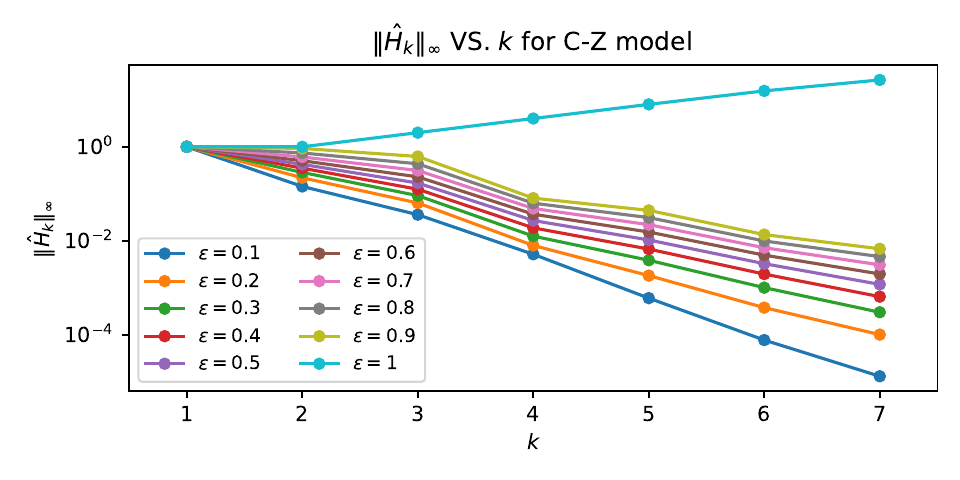}
      \caption{Third model: channel from \Eq{eq:Ch-CZ}}
      \label{fig:QC}
  \end{subfigure}
       \hfill
  \begin{subfigure}[r]{0.48\textwidth}
      \centering
      \includegraphics[width=\textwidth]{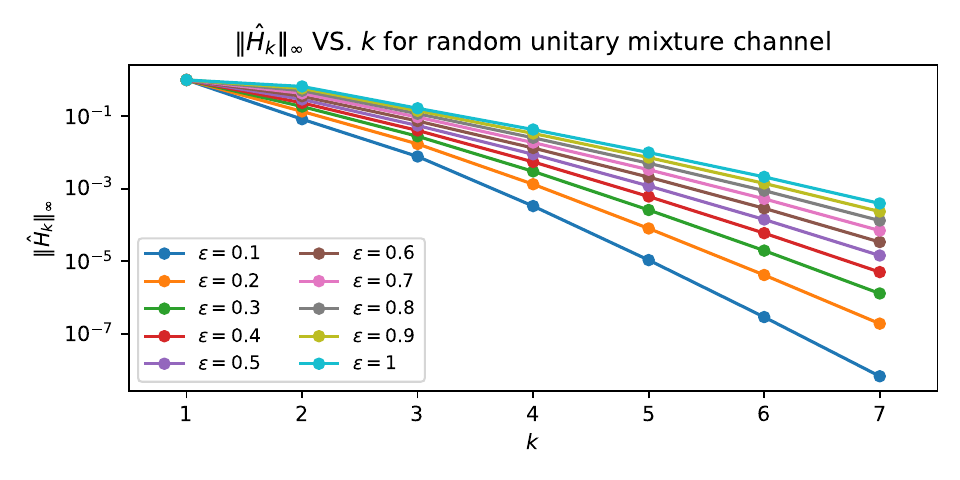}
      \caption{Second model: random unitary ($k=2)$}
      \label{fig:U2}
  \end{subfigure}
  \caption{Operator norms of $\hat H_k$ (see \Eq{eq:hHat}) as
    function of $k$ are displayed for simulations of four choices
    of $\mcE_\eps$ at $\eps = 0.1,0.2,\ldots,1$ on $10$ qubits. }
\end{figure}
The evident exponential decay of the norms for each
$\eps=0.1,0.2,\ldots,1$ suggests that $H_G(\epsilon)$ is indeed
quasi-local. The decay rates decrease (moderately) as $\epsilon$ is
increased. We note that there is a very little difference between
the decay at $\eps=1$ and the decay rates at $\eps<1$. We believe
this is due to the fact that the random correlators create a 
channel that is still dissipative, local, and ergodic, even when the
dissipators are absent. 

The second channel we simulated is identical to the first, except
that the correlators $\mcF_e$ taken to be random unitary channels:
$\mcF_e(\rho) = \frac{1}{k}\sum_{a=1}^k U^{(e)}_a\rho
U^{(e)\dagger}_a$, where $U^{(e)}_a$ are random unitaries taken from
the Haar measure on two qubits, and $k=1,2$. As in the previous
case, we verified robustness by running several such channels, which
all yielded similar results. Results of one instance with $k=1$ are
plotted in \Fig{fig:U}. The decays we see in \Fig{fig:U} resemble
the decays seen in \Fig{fig:Kraus}, but there is one notable
difference: Unlike in \Fig{fig:Kraus}, there is a significant
difference in the decay rate obtained from $\eps = 0.9$ and from
$\eps = 1$, where we see a substantial decrease in the decay rate
for $\eps=1$. We conjecture that in the absence of dissipators, when
there is a single unitary on each edge, ergodicity is
weak\footnote{We suspect that for $k=1$, the number of products of
unitaries required to achieve the full operator algebra is large.
(see Wielandt's inequality\cc{ref:Sanz2010primitivity,
ref:Angela2024genericquantum}). Therefore, the channel might have
eigenvalues that approach $1$, creating long range correlations
within the steady state.}, which might lead to a slow decay in 
$\norm{\hat H_k(\eps=1)}$. We confirm this by performing the same
simulation with $k=2$, i.e., a mixture of two unitaries on each
edge, which is likely to have stronger ergodicity. The results of a
single instance are presented in \Fig{fig:U2}. Note that a similar
decay profile is found in $\eps=1$ and $\eps=0.9$.

The third channel we studied the specific channel that was given in
\Eq{eq:Ch-CZ} in \Sec{sec:statement}. This is a
translationally-invariant channel, in which all the $W_i$ are given
by $W_i=0.9\ketbra{+}{+} + 0.1\ketbra{-}{-}$ and the correlators are
simply control-Z gates. The results of this simulation are plotted
in \Fig{fig:QC}. Here, at $\eps = 1$ we see the channel leads to
long-range Gibbs Hamiltonian $H_G$, which is not at all quasi-local.
In this case the channel reduces to $\frac{1}{10}\sum_i
\CZ_{i,i+1}\cdot\rho\cdot \CZ_{i,i+1}$ for which the steady state is
not unique (in-fact, any $\rho$ that is diagonal in the
computational basis is a fixed-point).  For the remaining choices of
$\eps$ values, one can see an exponential decay in agreement, which
is what we expect. It would be interesting to study the long-range
to short-range transition in $\eps$ for different choices of gates.
We leave this for future research. 

Overall, in our simulations we see that whenever a channel is of the
form given in \Thm{thm:quasi-local}, we obtain a quasi-local
Hamiltonian. In the limit $\eps = 1$, the channel may or may not
converge to a steady-state with quasi-local Gibbs Hamiltonian.

\section{Discussion and future research}
\label{sec:discussion}

In this paper, we established rigorous statements on the steady
states of a broad family of local stochastic dynamics,
defined as a convex sum of local channels~\eqref{def:E}. These are
channels that are composed of convex combination of single site
channels (named dissipators), and two-local channels (correlators),
for which the weight of the correlators is given by a parameter
$\eps\in[0,1)$. Our first main result is given in
\Thm{thm:quasi-local}, where we show that the steady state can be
written as a Gibbs state with a Gibbs Hamiltonian $H_G$ that admits
a series expansion in $\eps$. We proved that the $k$'th order in
$H_G$ is a $(k+1)$-geometrically local Hamiltonian $H_k$.  While we
were not able to derive an explicit upper bound on the norm of each
$H_k$, we used numerical simulations of several exemplary channels
of small systems, to show that the norm of $\eps^k H_k$ decays
exponentially in $k$ for all $\eps\in [0,1)$.

Although the statement of \Thm{thm:quasi-local} is intuitive, as the
steady state is reached by local dynamics, it is not a trivial
statement, as even in the classical regime local dynamical rule
might converge to a long range Gibbs
Hamiltonian\cc{ref:Kafri1998non-local}. The proof of the theorem 
follows an idea also used in \cRef{ref:Arad20131DAL} of moving to a
multivariate expansion, which gives rise to a short combinatorial
proof. We note that our proof is given in the language of quantum
channels, but its generalization to the more general setup of local
Lindbladians is straight-forward.

Below, we give several open questions and possible research
direction related to \Thm{thm:quasi-local}
\begin{enumerate}

  \item \textbf{Bounding the norm of $H_k$}. Arguably, the biggest
    missing piece in \Thm{thm:quasi-local} is a bound on the norm of
    $H_k$ terms. Showing that $\norm{H_k}\le \mu^k$ for some
    $\mu>0$, will rigorously imply that the Gibbs Hamiltonian is
    quasi-local for $\eps<1/\mu$. While we do not have a proof for
    that, we suspect that this is the case. We believe that the
    techniques used in the proof of \Thm{thm:avA}, where a similar
    threshold was established, can be used to prove a statement like
    that. We leave that for future research.

  \item \textbf{Generalization to general $2$-local channels:}
    Setting $\eps=1$ in~\eqref{def:E} leads to a general 2-local
	stochastic dynamics, which is not expected to converge to a
	quasi-local Gibbs state in all cases.  Nevertheless, it is often
	observed in the numerics that this is indeed the case. A
	necessary condition for this to happen is that the channel needs
	to be ergodic, as shown in \cRef{ref:Kafri1998non-local}, as
	well as in the numerical data in \Fig{fig:QC}.  One possible
	explanation for the convergence of general $2$-local channels is
	that in these cases, the general $2$-local channels can be
	self-decomposed into a sum of 1-local and 2-local channels. This
	way, an effective small parameter $\eps$ can be extracted from
	the dynamics itself.

  \item \textbf{Engineering Hamiltonians and Gibbs sampling:}
    An interesting future research direction is to understand the
    functional dependence of the quasi-quantum Gibbs Hamiltonian on
    the underlying stochastic dynamics. While \Thm{thm:quasi-local}
    gives us the general shape of the interaction graph of $H_G$, it
    does not directly provide the actual form of $H_k$.
    Understanding this dependence, we might be able to
    \emph{engineer} the local dynamics to converge to a specific,
    strictly local, Gibbs Hamiltonian (or an approximation of it).
    This would give rise to a natural family of quantum Gibbs
    preparation algorithms, which can be implemented on a quantum
    computer. Canceling high order terms in $H_G$ will probably
    require adding less local terms to stochastic dynamics. We
    conjecture that using quasi-local dynamics, one can reach a
    steady state that is a good approximation to a Gibbs state of
    strictly local Hamiltonian. This picture agrees with the recent
    Gibbs state preparation algorithms, that simulate a Lindbladian
    evolution with quasi-local generators\cc{ref:chen2023efficient,
    ref:guo2024designing}. 
    
    Additionally, if a large family of local quantum maps give rise
	to a manifold of Gibbs states, then one might use a
	\emph{variational} approach to find the local map that yields a
	particular Gibbs state. This approach might be particularly
	suitable for NISQ-type algorithms on noisy hardware, since some
	of the noise might be mitigated by the variational
	procedure\cc{ref:wang2021variational,
	ref:Consiglio2024variational, ref:ilin2024dissipative}.
\end{enumerate}

Our second main result is given in Theorem~\ref{thm:avA}, where we
computed the expectation values for local observables within the
steady state. For this purpose, instead of expanding the steady
state in $\eps$ (Schrodinger picture), we used the expansion of the
observable (Heisenberg picture). We showed that, just as for the
Gibbs Hamiltonian of the steady state, the expansion of the
observable in $\eps$ has a support that is growing with the order.
Moreover, in such case, we were able to give an upper-bound for the
norm of each term in the expansion. In other words, we found a
positive radius of convergence for the power series of
$A(\eps)=\sum_k A_k \eps^k$, which is an interesting result by
itself\cc{ref:Curt2024phases}. As a corollary, we estimated the
complexity of computing such expectation value \ within the radius
of convergence, which is summarized in Corollary~\ref{cor:alg}.

Following, are two remarks on the theorem:
\begin{enumerate}
  \item \textbf{Strengthening the results:}
    The assumption in Theorem~\ref{thm:avA} on
    the form of the dissipators $\mcD_i$ was taken for the sake of
    simplifying the result.
    This assumption can be relaxed to any ergodic dissipators which
    are diagonalizable, and for which the eigenoperators do not have
    large overlaps. Under such assumptions, the achieved bound will
    depend on the gap of the dissipators and the overlap
    between the eigenoperators. Moreover, we believe that
    the upper-bound on $\norm{A_k}$ can be strengthened by employing
    different techniques other than the $Q_1$ norm in the Wauli
    basis, which introduces an undesirable exponential dependence on the
    support of the observable $A$. Such
    strengthening will turn \Thm{thm:decay-of-corr} into a genuine
    `decay of correlations' statement.

  \item \textbf{Complexity thresholds:} 
    By the results of \cRef{ref:Verstrate2009DSE}, in the $\eps=1$
    limit, the steady state can be the ground state of any local
    frustration-free Hamiltonian, and therefore estimating
    expectation values there up to an error of $1/\poly(n)$ is a
    $\mathrm{QMA}_1$-hard problem. On the other hand, by
    \Thm{thm:avA} and Corollary~\ref{cor:alg}, for $\eps<\eps_0=1/e^6$,
    this problem is almost in P as it is quasi-polynomial,
    and therefore not expected to be QMA-hard (or even NP-hard).
    This implies the existence of a complexity crossover threshold
    somewhere in the range $[\eps_0, 1]$ and raises several
    questions: where exactly does this threshold lie? in particular,
    is it strictly smaller than $1$? if the answer is yes, then this
    might imply an interesting type of noise-resilience for such
    systems. It is also natural to ask if (and how) the threshold
    depends on underlying graph $G$. The fact that the bound
    $\eps_0=1/e^6$ in \Thm{thm:avA} is independent of $G$ is
    somewhat surprising, if we compare it other bounds of the same
    flavour, such as the Lieb-Robinson bound\cc{ref:Nechtergaele2009LR,ref:Kliesch2014LR}.
    Another interesting question is whether there exists a
    BQP-hard threshold between $\eps_0$ and the $\textrm{QMA}_1$
    threshold? A related question is how the mixing time of the
    system depends on $\eps$. Is there an $\tilde{\eps}\in (\eps_0,
    1)$ such that for $\eps<\tilde{\eps}$ the mixing time is
    polynomial in the system size, and therefore the problem of
    approximating local expectations in the steady state is in BQP?
    Finally, it would be interesting to understand how do all these different complexity regimes relate to
    the structural properties of $\rho_\infty$, e.g., correlation
    length, entanglement, efficient representability by a
    tensor-network, etc?

  \item \textbf{Finding a better algorithm for approximating 
    $\av{A}$}: A better algorithm than in Corollary~\ref{cor:alg}
    might be devised, reducing the quasi-polynomial time complexity
    to a polynomial time complexity, as we explain in the following.
    The quasi-polynomial time complexity in Corollary~\ref{cor:alg}
    is caused by the growth in support of $A_k$, which is contained
    in the $k$-radius ball around $\supp(A)$. Since we consider
    graphs that sit on a $D$-dimensional lattice, the size of such
    ball is proportional to $O(k^D)$, which leads to the $1/\delta$
    quasi-polynomial running time $e^{O(k^D)}$ when
    $k=\log(1/\delta)$. However, it can be seen from the proof of
    the theorem that $A_k$ is in fact a sum over $O(k)$-local Pauli
    (or Wauli) strings. Instead of summing all these strings, it
    might be possible to devise some sort of an
    importance-sampling algorithm for these strings (e.g., using
    ideas from \cRef{ref:Aharonov2023-noisy-sim}), thereby
    approximating the expectation value within the steady state. We
    leave this for future research.
    
\end{enumerate}

\section{Acknowledgments} 

We thank Yosi Avron, Yariv Kafri, Howard Nuer and Dorit Aharonov for
fruitful discussions.

This research project was supported by the National Research
Foundation, Singapore and A*STAR under its CQT Bridging Grant. 
It is also supported by the Israel Science Foundation (ISF)
under the Individual Research Grant No. 1778/17 of I.A.
and the joint Israel-Singapore NRF-ISF Research Grant No.
3528/20.
We acknowledge the support of the Helen Diller Quantum Center.

\appendix

\section{Proof of \Lem{lem:analyticality}}
\label{sec:analyticity-proof}

To prove \Lem{lem:analyticality}, we use the following fact on
quantum channels:
\begin{fact} [Ergodicity - proposition 3 of 
  \cRef{ref:Sanz2010primitivity}] \label{fact:ergodicity} A quantum
  channel $\mcE$ given in the Kraus form $\mcE(\rho)=\sum_{k=1}^m
  F_k \rho F_k^\dagger$ has a unique, full rank steady state
  $\sigma$ and no other eigenvalues for which $|\lambda|=1$ if and
  only if there exists $k\in \mathbb N$ such that the following
  vector space of operators
  \begin{align}
    S_k(\mcE) \EqDef \Span \{ F_{i_k}\cdot\ldots\cdot F_{i_1} 
      \;;\; i_1,\ldots,i_k \in \{ 1,\ldots ,m\}\}
  \end{align}
  spans the entire Hilbert space of operators $L(\mcH)$.
\end{fact}

To prove Bullet~\ref{bul1:anal}, we first consider the case when
$\eps_e=0$ for any $e\in E$, i.e., $\ueps=\uzero$. The channel is
then of the form
\begin{align*}
    \mcE_{\uzero}(\rho) 
      =\sum_{e=(i,j)\in E} \!\!\!\frac{p_e}{2} 
        [\mcD_i (\rho) + \mcD_j (\rho)] .
\end{align*}
Let $\mcD_i(\rho) = \sum_{k=1}^{m_i} D_k^{(i)} \rho
D_k^{(i)\dagger}$ be the Kraus decomposition of $\mcD_i$. The
Kraus operators of the full channel $\mcE_\uzero$ is the set (up
to proportionality constant of each operator)
\begin{align*}
    \bigcup_{i\in V}\{ D_k^{(i)}\;;\; k=1,\dots , m_i\} .
\end{align*}
It is easy to see that the steady state of $\mcE_{\uzero}$ is
unique, and given by the tensor product of the local steady states
of each local channel $\mcD_i$. Moreover, it is full rank, because
it is tensor product of full-rank operators, and $\mcE_\uzero$ has no
other eigenvalues with absolute value $1$. By
Fact~\ref{fact:ergodicity}, there is a $k\in\mathbb N$ such that
$S_k(\mcE_\uzero)=L(\mcH)$. From the functional form of the
channel in \Eq{eq:ch}, it is clear that for any $\ueps\in
[0,1)^n$, the set of Kraus operators of $\mcE_\ueps$ includes the
set $\bigcup_{i\in V}\{ D_k^{(i)}\;;\; k=1,\dots , m_i\}$,
again up to proportionality constants,
and therefore,
\begin{align*}
    S_k(\mcE_\uzero) \subseteq 
    S_k(\mcE_\ueps) .
\end{align*}
But as $S_k(\mcE_\uzero)=L(\mcH)$, we conclude that also
$S_k(\mcE_\ueps) = L(\mcH)$, and using Fact~\ref{fact:ergodicity},
uniqueness and full-rank of $\rho_{\infty}(\ueps)$ is guaranteed.

To prove the second bullet, note that the steady state satisfies
the following linear equation,
\begin{align}
\label{eq:steady-state-analytic}
  \big(\mcE_\ueps-\mcI\big) (\rho) = 0 .
\end{align}
Using standard basis vectorization, we map
$\rho_\infty(\ueps)=\sum_{ij}\rho_{ij}\ketbra{i}{j} \to \vec{v}\EqDef
\sum_{i,j}\rho_{ij}(\ueps)\ket{i}\ket{j}$, where we dropped the
explicit dependence of $\vec{v}$ on $\ueps$ in order to keep the
notation clean. Similarly, we map the channel $\mcE_\ueps$ to a
square matrix $A_\ueps$, which takes \Eq{eq:steady-state-analytic}
to the following linear matrix equation
\begin{align*}
    (A_\ueps - \Id)\vec{v}  = 0.
\end{align*}
Note that the matrix coefficients of $A_\ueps$ are affine in
$\eps_1,\ldots , \eps_n$.  
To solve this linear
equation, we use standard Gaussian elimination, which involves
subtracting from each row a multiple of another row, to reach a
row canonical form\cc{book:Kunze}. Since by Bullet~\ref{bul1:anal},
for any $\ueps \in [0,1)^n$ there is a unique steady state, we
conclude that the co-rank of $A-\Id$ must be $1$ (i.e., the rank
of its null space is $1$), and therefore, its canonical form must
look like
\begin{align*}
    (A-\Id)\mapsto B = 
    \left[\begin{array}{ccccc|ccccc}
        1 & 0 & 0 \ldots & 0 & f_1 
        &0 & 0 & \ldots & 0 \\
        0 & 1 & 0 \ldots & 0 & f_2 
        &0 & 0 &  & 0 \\
        & & \ddots & & 
        & &  & \ddots & \\
        0 & 0 & 0 \ldots & 1 & f_m 
        &0 & 0 & \ldots  & 0\\
        \hline 0 & 0 & 0 \ldots & 0 & 0 
        &1 & 0 & \ldots  & 0 \\
        0 & 0 & 0 \ldots & 0 & 0 
        &0 & 1 &  & 0 \\
        & & \vdots & & 
        & & & \ddots &  \\
        0 & 0 & 0 \ldots & 0 & 0 
        &0 & \ldots & & 1 \\
        0 & 0 & 0 \ldots & 0 & 0 
        &0 & \ldots & & 0 
    \end{array} \right]
\end{align*}
Since Gaussian elimination involves subtracting multiple of rows and
dividing each row by the first non-zero coefficient, we conclude
that $f_1,\ldots,f_m$ are in-fact polynomial quotient
\begin{align*}
    f_i(\eps_1,\ldots ,\eps_n)
    = \frac{h_i(\eps_1,\ldots,\eps_n)}{g_i(\eps_1,\ldots,\eps_n)}
\end{align*}
with $h_i,g_i$ being polynomials. As the system of equations
$(A-\Id)\vec{v} =0$ is equivalent to $B\vec{v}=0$, we
conclude that
\begin{align*}
    v_n & = v_{n-1}=\dots = v_{m+2} = 0 ,\\
    v_{m+1} & = c,\\
    v_i & = -c/f_i \quad \forall i=1,\ldots,m ,
\end{align*}
where $c$ is a constant that is fixed later by the normalization
condition $\Tr(\rho_\infty)=1$.  Note that $c$ is by itself a
polynomial quotient, thus also $v_i$ is, namely, $v_i=p_i/q_i$ for some
polynomials $p_i,q_i$ for each $i=1,\ldots,m$. Since the
coefficients of $\vec{v}$ are the entries of $\rho_\infty(\ueps)$,
and since we know that $\rho_\infty(\ueps)$ exists for any $\ueps\in
[0,1)^n$, we are guaranteed that there are no poles for any
$q_i$ in the open cube $[0,1)^n$, and therefore the entries of
$\rho_\infty(\ueps)$ are analytical in an open set containing
$[0,1)^n$\cc{ref:bochnak1997equivalence}.

\section{\texorpdfstring{$Q_1$}{Q1}-norm 
  and \texorpdfstring{$\mcT$}{T} properties}

Here we provide the proofs for properties of the $Q_1$-norm and the
super-operator $\mcT$ that were stated and used in the proof of
\Thm{thm:avA} \label{sec:Q1-proofs}
\begin{proof} [ of \Lem{lem:Q_1}] \ 
  \begin{enumerate}
    \item Recalling that in the dual Wauli
	  basis, $\tQ_0 = \Id$, the proof follows immediately.

    \item As $O\in S_k$, its dual Wauli
      expansion is $O=\sum_{|\ualpha|\le k} 
      c_\ualpha \tQ_\ualpha$. Therefore,
      \begin{align*}
        \norm{O}_{\infty} = \Big\| \sum_{|\ualpha|\le k} 
          c_\ualpha \tQ_\ualpha \Big\|_{\infty}
          \le \sum_{|\ualpha|\le k} |c_\ualpha|
            \cdot \norm{\tQ_\ualpha}_\infty 
		    \le (1+\lambda)^k\cdot \qnorm{O},
      \end{align*}
      where in the last step we used the fact that for a single
	  qubit dual-Wauli element $\tQ_\alpha$, the highest possible
	  eigenvalue (obtained for $\alpha = 3$) is $1+\lambda$, and
	  therefore the highest eigenvalue of $\tQ_\ualpha$ which is
	  supported on at most $k$ qubits is bounded by
	  $(1+\lambda)^k$. 
	  
	\item For inequality~\eqref{eq:Qinfineq-2}, 
	  we expand $O=\sum_\ualpha c_\ualpha \tQ_\ualpha$
	  and use the Hilbert-Schmidt product to express
	  $c_\ualpha = \Tr(Q_\ualpha\cdot
	  O)$. Then 
      \begin{align*}
        \qnorm{O} & = \sum_{\ualpha} |c_\alpha| = 
		  \sum_\ualpha |\Tr(Q_\ualpha \cdot O)| 
		  \le \sum_\ualpha \norm{Q_\ualpha}_1\cdot 
		    \norm{O}_\infty \le 4^\ell\cdot \norm{O}_\infty, 
      \end{align*}
	  where
	  in the first inequality we used 
	  Holder's inequality, and in the
	  last inequality we used the fact that the support of $O$ is
	  $\ell$, and therefore its expansion contains at most $4^\ell$
	  terms. In addition, for every $\ualpha$, it is easy to see
	  that $\norm{Q_\ualpha}_1=1$.
  \end{enumerate}
\end{proof}

To prove \Lem{lem:Q_1-growth} we need the
following fact: 
\begin{fact} \label{fact:norm}
  For a CPTP map $\mcE$ and an hermitian operator $O$ we have that 
  \begin{align*}
    \norm{\mcE^*(O)}_\infty \leq \norm{O}_\infty .
  \end{align*}
\end{fact}

\begin{proof} [ of fact \ref{fact:norm}]
  \begin{align*}
    \norm{\mcE^*(O)}_\infty = \max_{\rho\in D(\mcH)} |\Tr(\mcE^*(O)\rho)|
    = \max_{\rho\in D(\mcH)} |\Tr(O\mcE(\rho)) |
    \leq  \max_{\rho'\in D(\mcH)} |\Tr(O \rho'))| = \norm{O}_\infty ,
  \end{align*}
  where $D(\mcH)$
  denotes the set of quantum states, and in the first and last
  equalities we used the fact that for a hermitian operator
  $A$, $\norm{A}_\infty=\max_{\rho\in D(\mcH)} |\Tr(A\cdot\rho)|$.
\end{proof}

\begin{proof} [ of lemma \ref{lem:Q_1-growth}]
  Expand $O=\sum_\ualpha c_\ualpha \tQ_\ualpha$
  and decompose each string $\ualpha$ to $\ualpha_a,
  \ualpha_b$, where $\ualpha_a$ acts on spins in the support of $\mcE$
  and $\ualpha_b$ on the rest of the system. Then for every
  $\ualpha$, decompose $\tQ_\ualpha =
  \tQ_{\ualpha_a}\otimes \tQ_{\ualpha_b}$, and upper-bound
  \begin{align*}
    \norm{\mcE^*(O)}_{Q_1} &\le \sum_\ualpha
        |c_\ualpha| \cdot \norm{\mcE^*(\tQ_{\ualpha_a})\otimes 
		  \tQ_{\ualpha_b}}_{Q_1}
    \underbrace{=}_{\text{\eqref{eq:normOtimesQ}}} \sum_\ualpha
        |c_\ualpha| \cdot \norm{\mcE^*(\tQ_{\ualpha_a})}_{Q_1}
    \underbrace{\le}_{\text{\eqref{eq:Qinfineq-2}}} \sum_\ualpha
        |c_\ualpha|\cdot 4^k
		  \cdot \norm{\mcE^*(\tQ_{\ualpha_a})}_{\infty}  \\
        &\underbrace{\le}_{\text{Fact~\ref{fact:norm}}} 
	\sum_\ualpha |c_\ualpha|\cdot 4^k\cdot \norm{\tQ_{\ualpha_a}}_{\infty}
    \underbrace{\le}_{\text{\eqref{eq:Qinfineq-1}}} \sum_\ualpha
        |c_\ualpha|\cdot 4^k\cdot (1+\lambda)^k\cdot
		\underbrace{\norm{\tQ_{\ualpha_a}}_{Q_1}}_{=1} 
        =4^k(1+\lambda)^k\cdot \norm{O}_{Q_1} . 
  \end{align*}
  Finally, using the fact that $\lambda\le 1$, we conclude
  that $4^k(1+\lambda)^k \le 2^{3k}$.
\end{proof}

\begin{proof}[ of Claim~\ref{clm:tau}]
\label{sec:proof-Lemma-tau} 
  We start by deriving an explicit formula for the application of
  $\mcT$ on a dual Wauli operator $\tQ_\ualpha$. Applying $\mcE_1^*$
  on \Eq{eq:invadjK} we get: 
  \begin{align}
  \label{eq:TdualW0}
   \mcT (\tQ_\uzero) & = \mcE_1^*(\tQ_\uzero) = 0, \qquad \\
  \label{eq:TdualW}
   \mcT (\tQ_\ualpha) & = \frac{1}{q_\ualpha}\mcE_1^*(\tQ_\ualpha) =
    \frac{1}{q_\ualpha}\sum_{e\in E} p_e \mcF^*_e(\tQ_\ualpha)
      - \frac{1-q_\ualpha}{q_\ualpha}\tQ_\ualpha, 
	    \qquad (\ualpha \neq \uzero).
  \end{align}
  where in the first equation we have used the fact that dual quantum
  channels are unital, therefore $\mcE_1^*$, being a difference of
  such channels, maps the identity $\tQ_\uzero$ to $0$. Similarly,
  $\mcF_e^*$ is unital and since it is $2$-local (acting on an
  edge), we must have $|\supp(\mcF^*_e(\tQ_\ualpha))|\le
  1+|\supp(\tQ_\ualpha)|$. As the other terms of $\mcT$ cannot
  increase the support, we get that for any $O\in \mcS_k$, by
  expanding it in the dual-Waulis, we must have $\mcT(O) \in
  \mcS_{k+1}$ by the linearity of the channel. Moreover, given that
  $O$ is supported in a subset $A$, or alternatively, any dual-Wauli
  $\tQ_\alpha$ in its expansion is supported in $A$, under the same
  reasoning we would get that $\supp(\tQ_\alpha)$ is inside $A\cup
  e$ (provided that $e$ intersects $A$). This concludes
  Bullet~\ref{clm:tau:1}. 
  
  Next, let $E_0$ denote the subset of edges $e=(i,j)$ for which
  $\alpha_i=\alpha_j=0$, i.e., the edges that are fully
  outside the support of $\ualpha$, and let $E_1\EqDef E\setminus
  E_0$. For $e\in E_0$, we have $\mcF_e^*(\tQ_\ualpha) =
  \tQ_\ualpha$ since $\tQ_\ualpha$ reduces to the identity on qubits
  $i,j$. Thus, breaking the sum to $E_0$ and $E_1$ gives 
  \begin{align*}
    \sum_{e\in E}p_e \mcF^*_e(\tQ_\ualpha) 
    = \big(1- p_{E_1}\big)\tQ_\ualpha
     + \sum_{e\in E_1}p_e\mcF^*_e(\tQ_\ualpha),
  \end{align*}
  where we use the notation $p_{E_1}\EqDef \sum_{e\in E_1}p_e$ and
  the fact that $\{p_e\}$ form a probability distribution.
  Plugging this into \Eq{eq:TdualW} for $\ualpha \neq \uzero$ get
  \begin{align} \label{eq:tau-tQ_a}
    \mcT(\tQ_\ualpha) &= \frac{1}{q_\ualpha}
      \sum_{e\in E_1} p_e\mcF^*_e(\tQ_\ualpha) 
     + \left(1-\frac{p_{E_1}}{q_\ualpha}\right)
	     \tQ_\ualpha .
  \end{align}
  With this expression in hand, we proceed to prove
  Bullet~\ref{clm:tau:2}. Let $O$ be a general observable expanded
  in the dual-Waulis $O=\sum_\ualpha c_\ualpha \tQ_\ualpha$. Before
  bounding $\qnorm{\mcT(O)}$, we note following:
  \begin{enumerate}[i]
    \item For any dual Wauli $\tQ_\ualpha$, clearly 
      $\qnorm{\tQ_\ualpha}=1$
      
    \item Therefore, by \Lem{lem:Q_1-growth},
      $\qnorm{\mcF^*_e(\tQ_\ualpha)} \le 2^6 \qnorm{\tQ_\ualpha} = 64$
      as $\mcF^*_e$, is $2$-local.
      
    \item Given the relation between $q_i$ and $p_e$ from 
      \Eq{eq:simpler-channel}, it follows that $p_{E_1}\le 2
      q_{\ualpha}$, since:
      \begin{align*}
          p_{E_1} = \sum_{e\in E_1} p_e 
            \le \sum_{i\in \supp(\ualpha)}\sum_{e\ni i} p_e
          =2\sum_{i\in \supp(\ualpha)}q_i = 2 q_{\ualpha}
      \end{align*}
  \end{enumerate}
  Using these, together with \eqref{eq:TdualW0} and
  \eqref{eq:tau-tQ_a}, we obtain the desired bound: 
  \begin{align*}
    \norm{\mcT(O)}_{Q_1} &\le \sum_{\ualpha} |c_\ualpha|\cdot 
	  \norm{\mcT(\tQ_\ualpha)}_{Q_1} \\
    &\le \sum_{\ualpha \neq \uzero} |c_\ualpha|\cdot 
      \left(\frac{1}{q_\ualpha}\sum_{e\in E_1}
	      p_e\norm{\mcF^*_e(\tQ_\ualpha)}_{Q_1} +
           \left|1-\frac{p_{E_1}}{q_\ualpha}\right|\cdot
		     \norm{\tQ_\ualpha}_{Q_1} \right)\\
    &\le \sum_{\ualpha \neq \uzero} |c_\ualpha| 
    \left( \frac{p_{E_1}}{q_\ualpha}\cdot 64 +
      \left|1-\frac{p_{E_1}}{q_\ualpha}\right| \right) \\
     & \le 130\cdot \sum_{\ualpha} |c_\ualpha| = 130\cdot \qnorm{O} .
  \end{align*}
  
\end{proof}

\bibliographystyle{ieeetr}
\bibliography{KGB}

\end{document}